%% file: paper_arxiv.tex
\documentclass[11pt]{article}

\usepackage[margin=1in]{geometry}
\usepackage[utf8]{inputenc}
\usepackage[T1]{fontenc}
\usepackage{microtype}
\usepackage[hidelinks]{hyperref}
\usepackage{graphicx}
\usepackage{booktabs}
\usepackage{amsmath}
\usepackage{amssymb}
\usepackage{amsthm}
\usepackage{mathtools}
\usepackage{xcolor}
\usepackage{tikz}
\usetikzlibrary{arrows.meta,positioning,shapes.geometric,calc}

\newtheorem{definition}{Definition}
\newtheorem{example}{Example}
\newtheorem{invariant}{Invariant}
\makeatletter
\@ifundefined{theorem}{\newtheorem{theorem}{Theorem}}{}
\@ifundefined{corollary}{\newtheorem{corollary}{Corollary}}{}
\makeatother

\newcommand{\cmark}{$\checkmark$}
\newcommand{\xmark}{$-$}

\providecommand{\sysname}{\textsf{cryptitalk-wallet}}

\title{Safety Invariants for Agents Orchestrating Irreversible State Transitions\\
\large A Four-Dimensional Formalism Evaluated on Public Ledgers}
\author{%
  Zhaoming Yin \\
  \small Datumpont AI LLC, California, USA \\
  \small \texttt{zhaoming@datumpont.ai} \\
  \small \href{https://datumpont.ai}{datumpont.ai}
}
\date{\today}

\begin{document}
\maketitle

\input{content/abstract}

\input{content/01_introduction}
\input{content/02_background}
\input{content/03_formalism}
\input{content/04_orchestration}
\input{content/05_evaluation}
\input{content/06_related_work}
\input{content/07_discussion}
\input{content/08_conclusion}

\bibliographystyle{plain}
\bibliography{references}

\end{document}

%% file: content/abstract.tex
\begin{abstract}
Autonomous agents are increasingly asked to produce
\emph{irreversible} effects on external systems---transferring
funds, writing to durable storage, actuating hardware. Existing
agent frameworks (ReAct, Reflexion, MCP, function calling)
optimize \emph{task success} on benchmarks and give little
attention to the safety of irreversible side-effects. We formalize
one such setting, movement of value across public ledgers, as
state transitions in a four-dimensional space indexed by
$(\textsf{wallet}, \textsf{chain}, \textsf{address},
\textsf{protocol})$, and use that formalism to state and prove a
guarantee we call \emph{execution fidelity}: under a fault model
admitting planner mis-mapping, ambiguous outcomes, retries,
at-least-once delivery, and delegated non-human callers, a
session's realized effect on the ledger is either nothing at all
or exactly the transition that was rendered to the user, exactly
once. The theorem deliberately does not claim that the rendered
transition matches the user's intent---no runtime layer can decide
that---but it confines that unbounded question to a single
predicate over a finite object, which is what makes a preview a
sufficient control rather than a formality. Seven safety
invariants, derived from the fidelity condition rather than
enumerated from experience, discharge the guarantee; we give their
enforcement points and the three places the guarantee stops.
Empirically, on a controlled $N{=}60$ adversarial suite the stack
lifts pass rate by $\sim$74 percentage points over a naive-ReAct
baseline on two write-aggressive backing models, but by only
$\sim$3 points on a write-cautious one---evidence that
single-model evaluations of agent safety stacks are close to
unfalsifiable, and an argument for treating the backing model as a
first-class experimental variable. The system is deployed; 108
production write operations across 8 chains and 8 transition
primitives back the failure taxonomy. Although the evaluation
setting is public ledgers, the formalism and invariants apply to
any probabilistic agent acting on irreversible external state.
\end{abstract}

%% file: content/01_introduction.tex
\section{Introduction}
\label{sec:intro}

Autonomous agents built on probabilistic planners are increasingly
asked to produce \emph{irreversible} effects on external
systems---signing on-chain transactions, writing to durable
storage, actuating hardware. A user who asks ``deposit all my
USDC to Hyperliquid'' or ``swap 200 USDT on Ethereum to BTC''
expects the agent to navigate chain selection, contract
approvals, bridge providers, and RPC consistency on their behalf;
naive wrapping of signing and broadcast tools inside such an
agent reliably produces money-burning failure modes.

The question this paper answers is what can be \emph{guaranteed}
in that setting. The tempting answer---that the agent does what
the user meant---is not available: deciding whether a transaction
faithfully renders a natural-language request is not something a
runtime layer can do, and any paper claiming otherwise is
overselling. The answer we can defend is narrower and, we argue,
more useful. Call a session's realized effect on the ledger
\emph{faithful} if it is either nothing at all or exactly the
transition that was rendered to the user, exactly once. We prove
that a deployed agent system satisfies this property---we call it
\emph{execution fidelity}---under a fault model that admits
planner mis-mapping, ambiguous outcomes, retries, at-least-once
delivery, and delegated non-human callers
(\S\ref{sec:fidelity-theorem}).

What that buys is a reduction. Intent-mapping remains unsolved,
but it becomes the \emph{only} unsolved thing: every other way a
session can go wrong is forced to produce no effect at all. An
unbounded semantic question is thereby confined to a bounded
one---does this four-coordinate tuple and amount match what was
asked---which a human can answer in a second and which a scope
check can approximate when no human is present. The
four-dimensional formalism exists to make that comparison
decidable rather than to describe wallets elegantly.

The problem is not hypothetical. We observe three failure classes
in the wild that the guarantee is built against. RPC-level:
\emph{phantom success}, where a verify-on-broadcast check returns
a hash the RPC retracts seconds later, and the agent reports
\emph{failed} on a transaction that actually landed. LLM-level:
\emph{stale-context drift}, where an agent that concluded ``the
user has no Hyperliquid account'' parrots that conclusion on later
turns even after the cause is fixed, over-weighting a prior
observation rather than re-querying ground truth.
Exchange-API-level: \emph{documented-atomicity drift}, where an
action specified as atomic cancel-and-replace applies only the
cancel half on a replacement-side rejection, leaving the user with
no resting order. One of these cost real money: a retry after a
misreported failure produced a duplicated 200~USDT swap
(\S\ref{sec:evaluation}).

A second finding concerns how such systems should be evaluated at
all. On a controlled $N{=}60$ adversarial suite the invariant
stack lifts pass rate by $\sim$74 percentage points over a
naive-ReAct baseline on two write-aggressive backing models---and
by $\sim$3 points on a write-cautious one, which declines write
tools on ambiguous prompts unprompted. The measured value of a
safety layer is thus largely a function of which model sits
underneath it, so a single-model evaluation can report almost any
number it likes. We treat the backing model as a first-class
experimental variable for this reason.

Concretely, this paper contributes the following. We formalize the
wallet problem as state transitions in a four-dimensional space
$\mathcal{S} = W \times C \times A \times P$, together with a
fault model over the agent--executor surface that is separate from
the usual adversary model over the ledger
(\S\ref{sec:formalism}). We state and prove the fidelity theorem,
and \emph{derive} seven safety invariants from the fidelity
condition by case analysis rather than enumerating them from
experience, giving each one's enforcement point and the three
places the guarantee stops
(\S\ref{sec:orchestration}). We describe \sysname{}, a deployed
multi-tier agent system---executor, avatar, and analyst tiers
glued by an agent SDK over MCP and cross-tier delegation
tokens---and report its on-chain track record of 108 production
write operations across 8 source chains, transactions that landed
on a public ledger rather than in a simulator, from which the
failure taxonomy is drawn. Finally we report a controlled
adversarial suite at $N{=}60$ across four backing models, whose
cross-model spread is itself the methodological finding above
(\S\ref{sec:evaluation}).

The remainder of the paper is structured as follows.
\S\ref{sec:background} surveys agent frameworks and on-chain
safety primitives. \S\ref{sec:formalism} introduces the
four-dimensional state space, the fault model, and the fidelity
condition. \S\ref{sec:orchestration} derives the invariants from
that condition, gives their enforcement points, and proves the
theorem. \S\ref{sec:evaluation} reports the on-chain track record
and the adversarial-suite results. \S\ref{sec:related} situates the work
against prior agent and wallet literature. \S\ref{sec:discussion}
discusses limitations and \S\ref{sec:conclusion} concludes.

%% file: content/02_background.tex
\section{Background}
\label{sec:background}

This section grounds the formalism of \S\ref{sec:formalism} for
readers arriving from any of three communities the paper sits
between: agent / LLM tool-use research, public-ledger systems and
DeFi, and applied security.

\paragraph{Hierarchical deterministic wallets and custody.}
A single mnemonic seed (BIP-39) derives, via BIP-32 / BIP-44, an
unbounded family of key pairs and therefore addresses; authority
over funds at any derived address reduces to control of the seed.
In \sysname{} the seed lives in encrypted-at-rest storage keyed
off a master key held in a cloud key-management service (KMS),
and signing happens in a short-lived
in-process derivation; no plaintext private key is logged,
returned over the API, or persisted. The orchestration-layer
arguments of \S\ref{sec:orchestration} are independent of this
custody choice---the seven invariants apply identically over
hardware-security-module, threshold-signature (MPC), and
smart-contract-wallet substrates, because each is a property of
the agent--executor surface, not of the signing one. Key-recovery
questions (lost mnemonic, social-recovery designs, smart-wallet
guardian sets~\cite{erc4337,safe_aa_modules,argent_session_keys})
are orthogonal and out of scope.

\paragraph{Chains, protocols, tokens.}
Each chain imposes its own address format (EVM hex, Bitcoin bech32
/ taproot, Solana base58, Cosmos bech32). On a given chain, a
user's address holds value under various \emph{protocols}: the
native coin, fungible tokens (ERC-20, SPL / Token-2022, IBC
denominations), LP positions, wrapped or staked derivatives,
bridge-custodied balances, off-chain exchange balances
(Hyperliquid spot and perp), and outcome shares on
CTF-Exchange-style prediction-market venues. We treat protocol as
a dimension orthogonal to address because neither coordinate
determines the other: one address holds balances under many
protocols simultaneously, and one protocol exists at many
addresses across many holders. Collapsing the two---modelling a
position as ``an address that happens to hold token $X$''---loses
the ability to express a transition that changes what is held
while holding where it is held fixed, which is exactly a
\textsc{Swap}. Operation surfaces in
production cluster into a small number of categories: native
transfers, single-chain DEX swaps, cross-chain bridges (THORChain,
Squid / Axelar GMP, Skip on Cosmos), exchange operations,
prediction-market trades, and arbitrary EVM contract calls.
\S\ref{sec:formalism} formalizes these as transition primitives
on a single state space.

\paragraph{Agent orchestration substrate.}
We instantiate the agent with a large language model invoked
through \texttt{claude-agent-sdk}~\cite{mcp2024}, but the
formalism and invariants do not depend on that choice. Tool calls
cross a structured boundary defined by the Model Context Protocol
(MCP)~\cite{mcp2024}: the agent emits a JSON-shaped
\texttt{tool\_use} payload naming a tool and its arguments, the
runtime executes the tool and returns a structured result, and
the agent decides whether to continue. Two caller types appear
throughout: a \emph{human user} driving the agent interactively,
and an \emph{avatar}---a persistent agent instance, installed and
parameterized by the user, that runs on a schedule or trigger and
acts on the user's behalf while the user is not present. Avatars
hold no keys and reach no external system directly; they call the
same API a human calls, under a scoped credential
(\S\ref{sec:multi-tier}). Prior work on LLM agents
relevant here includes ReAct~\cite{yao2023react} and
Reflexion~\cite{shinn2023reflexion} (single-process planners),
AutoGen~\cite{wu2023autogen} and MetaGPT~\cite{hong2023metagpt}
(multi-agent orchestration), and DecoAgent~\cite{decoagent}
(decomposition under tool constraints).

\paragraph{Threat model.}
Public-ledger transactions are irreversible: once included in a
finalized block, no rollback path exists short of chain reorg,
which we treat as out of scope. Authority on a user's funds
reduces to a signature from a key derivable from the user's seed.
We treat the user-side software stack---LLM agent, MCP client,
wallet service---as \emph{fallible} (may misinterpret intent,
retry on ambiguous results, mis-marshal a boolean) but not
directly \emph{adversarial}, and assume an honest-majority network
and honest block explorers for post-broadcast verification. The
threat-surface taxonomy we work against:

\begin{description}\setlength\itemsep{2pt}
  \item[In scope] agent-side fallibility (the seven invariants of
    \S\ref{sec:orchestration}); cross-tier authority leakage
    addressed by scoped delegation tokens; HTTP-boundary
    duplicate-delivery races addressed by the idempotent
    receiver.
  \item[Partly in scope] indirect prompt
    injection~\cite{mcp_sok} through avatar-ingested feeds (RSS,
    social, third-party webhooks)---scope-bounded by
    Invariant~\ref{inv:delegation} but not eliminated;
    compromised-avatar single-point-of-compromise; MEV
    (sandwich, JIT, oracle/bridge front-running)
    \cite{daian2020flash,qin2022quantifying_mev,zhou2023sok}
    addressed in part by limit-order discipline on Trade
    primitives, not by any of the seven invariants.
  \item[Out of scope] signature-key compromise; smart-contract
    bugs in routed protocols; Sybil attacks on consensus;
    cross-chain bridge protocol vulnerabilities~\cite{sok_bridges}.
\end{description}

The seven invariants in \S\ref{sec:orchestration} address
agent-side and receiver-side failures; the partly-in-scope and
out-of-scope surfaces are revisited in \S\ref{sec:discussion}.
Recent work on LLM context fragility~\cite{liu2023lost,
laban2025llms, lost_beginning_reasoning} and on agent-layer attack
surfaces in MCP systems~\cite{mcp_sok} complements our
fallible-software framing rather than substituting for it.

%% file: content/03_formalism.tex
\section{A Four-Dimensional State Space}
\label{sec:formalism}

\input{content/figures/4d_space}

\begin{definition}[Wallet state]
Let $W$ be the set of mnemonics the user custodies, $C$ the set
of chains, $A_c$ the set of addresses derivable on chain $c \in
C$, and $P_c$ the set of protocols (tokens, contracts, off-chain
accounts) recognized on chain $c$. A \emph{state} $s \in
\mathcal{S}$ is a function
\[
s : W \times C \times A \times P \to \mathbb{Z}_{\ge 0},
\]
where $s(w, c, a, p)$ denotes the balance under protocol $p$ at
address $a$ on chain $c$, owned by mnemonic $w$. Balances are
non-negative integers in the protocol's smallest unit (wei,
satoshi, lamport, or the protocol's specified decimals); the
human-decimal projection (e.g.\ \texttt{54.33 USDC}) is a UI
abstraction over the underlying integer.
\end{definition}

A wallet's observable balance is a finite-support slice of $s$:
the tool returns only $s(w, c, a, p) > 0$ entries.

\begin{definition}[Transition]\label{def:transition}
A \emph{transition} is a partial function $\tau : \mathcal{S} \to
\mathcal{S}$ that modifies the balances at a bounded set of $(w,
c, a, p)$ quadruples. It is partial because a transition has
preconditions: we write $\mathrm{dom}(\tau)$ for the set of
states on which $\tau$ is defined---those whose pre-state carries
sufficient balance at every coordinate $\tau$ decrements, together
with whatever chain-specific preconditions the primitive requires
(a live ERC-20 approval, an unspent UTXO set covering the output,
an exchange account in a tradeable status). The distinction is
load-bearing rather than notational: $s \notin
\mathrm{dom}(\tau)$ is precisely the condition
Invariant~\ref{inv:preflight} (\S\ref{sec:orchestration}) checks
before signing, and the
reason it must be checked against ground-truth chain state rather
than against a cached read is that $\mathrm{dom}(\tau)$ is
evaluated at broadcast time, not at planning time. Writing $s' =
\tau(s)$ for $s \in \mathrm{dom}(\tau)$, every transition
satisfies two formal constraints.

\medskip\noindent\textbf{Conservation.} A \emph{conservation class} is a set
$K \subseteq W \times C \times A \times P$ over which the
protocol's specification guarantees aggregate invariance (e.g.\
\emph{global USDC supply}, fixed by the issuer's mint contract;
or \emph{a single mnemonic's Hyperliquid spot sub-account}, fixed
by the exchange's reconciliation rule). Each class has a fee
allowance $\varepsilon_K(\tau) \ge 0$ and a mint/burn delta
$\delta_K(\tau) \in \mathbb{Z}$ such that
\makeatletter
\if@twocolumn
\begin{multline*}
\sum_{(w,c,a,p) \in K} s'(w,c,a,p) \\
= \sum_{(w,c,a,p) \in K} s(w,c,a,p)
\;-\; \varepsilon_K(\tau) + \delta_K(\tau).
\end{multline*}
\else
\[
\sum_{(w,c,a,p) \in K} s'(w,c,a,p)
= \sum_{(w,c,a,p) \in K} s(w,c,a,p) - \varepsilon_K(\tau) + \delta_K(\tau).
\]
\fi
\makeatother
Per-chain bridging via lock-and-mint, for example, has $\delta_K
\neq 0$ on each individual chain's USDC class but $\delta_K = 0$
on the global USDC class.

\medskip\noindent\textbf{Signature authorization.} For every coordinate where
the balance strictly decreases, the transition carries a valid
signature from a key derivable from $w$:
\makeatletter
\if@twocolumn
\begin{multline*}
\forall (w,c,a,p):\; s'(w,c,a,p) < s(w,c,a,p) \\
\Longrightarrow\; \exists \sigma:\; \mathrm{Verify}_{\mathrm{pk}(w,c,a)}(\tau, \sigma) = 1,
\end{multline*}
\else
\[
\forall (w,c,a,p):\; s'(w,c,a,p) < s(w,c,a,p) \;\Longrightarrow\; \exists \sigma:\; \mathrm{Verify}_{\mathrm{pk}(w,c,a)}(\tau, \sigma) = 1,
\]
\fi
\makeatother
where $\mathrm{pk}(w,c,a)$ is the derivation of the public key
for chain $c$ and address $a$ from mnemonic $w$
(\S\ref{sec:background}), and $\mathrm{Verify}$ is the chain's
signature-verification primitive applied to a canonical encoding
of $\tau$.

Both conditions are \emph{well-formedness} properties of the
substrate, not guarantees the orchestration layer supplies. Any
transition a ledger accepts already satisfies them, and they
constrain the model rather than the agent: they say that a
transition cannot invent or destroy value outside a declared
mint/burn, and cannot move value out of a coordinate without the
corresponding key. What they conspicuously do \emph{not} say is
that the transition carrying that signature is the one the user
asked for---a hundred unintended transfers signed by the correct
key satisfy both conditions. Closing that distance is the job of
\S\ref{sec:orchestration}, and the reason the formalism is worth
stating is that it makes the target precise: the orchestration
layer must ensure the tuple that gets signed is the tuple that
was promised.
\end{definition}

We identify eight transition primitives, matching the categories
observed in production (\S\ref{sec:evaluation},
Table~\ref{tab:primitives}). The point of tabulating them is to
establish that the state space needs all four axes and no fewer:
each of $C$, $A$ and $P$ is moved by some primitive while the
others are held fixed, so no axis is a function of the rest and
none can be dropped without losing the ability to express an
operation the system actually performs. For each primitive we
list which dimensions change (\cmark) and which are held fixed
(\xmark):

\begin{center}
\small
\begin{tabular}{lcccc}
\toprule
Primitive & Wallet & Chain & Address & Protocol \\
\midrule
\textsc{Transfer}        & -- & --     & \cmark & --     \\
\textsc{Swap}            & -- & --     & --     & \cmark \\
\textsc{Bridge}          & -- & \cmark & \cmark & $\pm$  \\
\textsc{Deposit}         & -- & \cmark & \cmark & \cmark \\
\textsc{Withdraw}        & -- & \cmark & \cmark & \cmark \\
\textsc{Trade}           & -- & --     & --     & \cmark \\
\textsc{Intra-Exchange}  & -- & --     & --     & \cmark \\
\textsc{Contract}        & -- & --     & --     & $\pm$  \\
\bottomrule
\end{tabular}
\end{center}

\textsc{Intra-Exchange} (e.g.\ moving balance between a
Hyperliquid spot and perp sub-account) holds chain and address
fixed while crossing the protocol axis. \textsc{Contract} (an
arbitrary EVM call) holds chain and address fixed and may or may
not cross the protocol axis depending on the call's effect
(\texttt{WETH.withdraw} converts WETH$\to$ETH, crossing protocol;
a read-only call holds all four fixed and is not a transition).

The wallet column deserves comment, since no primitive moves it
and a reader may reasonably ask why $W$ is in the model at all.
It is there precisely \emph{because} it is fixed: the signature
condition of Definition~\ref{def:transition} pins $w$ across a
transition, so a transition whose source coordinate names the
wrong mnemonic is not a different-but-valid operation---it is one
the user's key cannot authorize. Users routinely custody several
mnemonics, and selecting the wrong one is failure class~4
(wrong-wallet selection, \S\ref{sec:failure-classes}), our
second-most-common observed class. A three-axis model would have
no coordinate in which to express that error: the operation would
look well-formed and simply fail at the balance check for
reasons the model could not name. Carrying $W$ explicitly is what
turns a confusing insufficient-funds error into a
wrong-wallet diagnosis, which is the difference between a user
who retries blindly and one who switches wallets.

\begin{definition}[Composition and plans]\label{def:composition}
The \emph{composition} $\tau_2 \circ \tau_1$ of two transitions
is the partial function $s \mapsto \tau_2(\tau_1(s))$, defined
exactly when $s \in \mathrm{dom}(\tau_1)$ \emph{and} $\tau_1(s)
\in \mathrm{dom}(\tau_2)$. A finite chain $\pi = \tau_n \circ
\cdots \circ \tau_1$ is a \emph{plan}; the conservation and
signature-authorization conditions of
Definition~\ref{def:transition} compose coordinate-wise.

\emph{Cross-chain bridges} are accommodated within this
single-$\mathcal{S}$ framing: \textsc{Bridge} affects per-chain
conservation classes on \emph{both} the source chain (lock or
burn, $\delta_K < 0$) and the destination chain (mint or release,
$\delta_K > 0$), while the global class for the bridged asset has
$\delta_K = 0$. Plan composition across chains is therefore just
ordinary composition in $\mathcal{S}$ with different
conservation classes contributing on each step.

Plans introduce a separate, runtime concern not visible at the
level of Definition~\ref{def:transition}: an executor that
completes $\tau_i$ but fails before $\tau_{i+1}$ leaves the chain
in an intermediate state that is reachable but not the one the
plan intends. Avatar-tier discipline against re-issuing $\tau_i$
on a transient failure of $\tau_{i+1}$ is formalized in
Invariant~\ref{inv:plan-retry}.
\end{definition}

\begin{example}[``Deposit all my USDC to Hyperliquid'']
Mapping the user intent to a transition tuple:
\begin{itemize}
  \item Source: $s(w, \text{arbitrum}, a_w^{\text{EVM}},
        \text{USDC}) = 54{,}331{,}159$ (54.33 USDC in 6-decimal
        units).
  \item Transition: \textsc{Deposit} with parameters
        $(\text{arbitrum}, a_w^{\text{EVM}}, \text{USDC})
        \mapsto (\text{hyperliquid-L1}, a_w^{\text{HL}},
        \text{perp-USDC})$.
  \item Implementation: the existing $\textsf{transfer}$ tool
        issues an ERC-20 call $\text{USDC}.\text{transfer}(b,
        54{,}331{,}159)$ to the bridge contract address $b$; the
        bridge observes the transfer and the L1 validator set
        credits \texttt{perp-USDC} to $a_w^{\text{HL}}$.
\end{itemize}
\end{example}

\begin{example}[``Swap 200 USDT Ethereum to BTC'' via THORChain]
Source $s(w, \text{ethereum}, a_w^{\text{EVM}}, \text{USDT}) =
2{,}087 \cdot 10^6$. Transition: \textsc{Bridge} with chain and
protocol both changing, executed via a single call to the
THORChain router's $\texttt{depositWithExpiry}$ on Ethereum; the
bridge observes, atomically swaps USDT$\to$RUNE$\to$BTC across
its liquidity pools, and emits a BTC send to the specified
destination address.
\end{example}

Both examples were chosen for clarity, and clarity is exactly what
makes them unrepresentative: each names its source, asset, amount
and destination, so the mapping to a tuple is mechanical. Real
requests are frequently not like this, and a formalism that only
handles the easy case would be worth little. The interesting
question is what the model says about a request that does
\emph{not} determine a transition.

\begin{example}[``Move my stables somewhere safer'']
\label{ex:ambiguous}
This request fixes almost nothing. \emph{Which} stablecoins---the
user may hold USDC, USDT and DAI across several chains, so the
protocol coordinate $p$ is a set, not a value. \emph{Safer} in
what sense: fewer bridge hops, a chain with faster finality, a
non-custodial venue, or simply out of a protocol the user has read
something about? The destination coordinates $(c, a)$ are
underdetermined, and different readings select genuinely different
transitions with different costs and different risks.

The formalism's response is not to guess well. It is that no
single $\tau$ has been determined, so no preview can be rendered,
so the confirm gate of \S\ref{sec:orchestration} cannot be
satisfied and no key is derived. The correct behavior is a
clarification request. This is worth stating explicitly because it
is a case where the right answer is to \emph{not act}, and where
an agent optimized for task-completion benchmarks is under
precisely the wrong incentive: picking a plausible reading scores
as success, and refusing scores as failure, while on a public
ledger the ranking is reversed.
\end{example}

This is also the principled basis for the scoring rule in
\S\ref{sec:methodology}, which counts a conservative refusal or a
clarification request as a pass whenever the expected outcome is
non-\texttt{succeed}. Without Example~\ref{ex:ambiguous} that rule
reads as scorer leniency; with it, it is the only scoring that
matches what the system is trying to do.

\subsection{Sessions, faults, and what can be guaranteed}
\label{sec:fidelity}

The machinery above describes what \emph{can} happen on a ledger. It
does not yet say what a probabilistic agent driving that ledger
should be held to. We state that target here, before describing any
mechanism, because the guarantee we can honestly offer is narrower
than a reader might expect and the narrowing is the substance of the
claim rather than a caveat on it.

\begin{definition}[Session]\label{def:session}
A \emph{session} $\sigma$ is one pass through the pipeline: a user
request $u$ expressed in natural language, a single transition
$\tau_p$ that the agent commits to and renders before acting---the
\emph{preview}---and whatever execution follows. We write
$\mathrm{real}(\sigma)$ for the multiset of transitions that take
effect on the ledger as a consequence of $\sigma$. The multiset,
rather than set, is deliberate: a transition applied twice is the
failure mode of \S\ref{sec:evaluation}'s realized double-spend, and
a formalism that silently collapses it cannot express the property
we need.
\end{definition}

\medskip\noindent\textbf{Fault model.} \S\ref{sec:background}'s
threat model bounds the \emph{adversary}. What follows bounds the
\emph{fallibility} of non-adversarial components, which is a
different axis and the one this paper is about. Within a session we
admit:

\begin{description}\setlength\itemsep{2pt}
  \item[F1 (intent mis-mapping)] The planner may commit to any
    $\tau_p$ whatsoever, with arbitrary relation to $u$. This fault
    is \emph{unbounded}: nothing in a runtime layer can decide
    whether a tuple faithfully renders a natural-language request.
  \item[F2 (pre-state drift)] Between planning and broadcast the
    chain may move, so that $s \notin \mathrm{dom}(\tau_p)$ at the
    moment of execution even though it held at planning time.
  \item[F3 (outcome ambiguity)] After broadcast, whether $\tau_p$
    took effect may be unobservable within any bounded window, or
    observably wrong---an RPC may return a hash it later retracts.
  \item[F4 (re-issue)] The planner may invoke any tool any number of
    times, in particular after an error or an ambiguous result.
  \item[F5 (duplicate delivery)] A channel carrying externally
    originated intent is at-least-once; the same request may arrive
    more than once.
  \item[F6 (non-human caller)] A call may originate from an
    autonomous caller acting under delegated authority rather than
    from the user directly.
\end{description}

We assume throughout that executor code is correct and its checks
are not bypassable, that signing keys are uncompromised, that the
chain does not reorganize, and that the RPC is honest---the last
already flagged as a trust assumption in
Invariant~\ref{inv:preflight}. These are assumptions, not results.

\begin{definition}[Execution fidelity]\label{def:fidelity}
A session $\sigma$ satisfies \emph{execution fidelity} if
\[
\mathrm{real}(\sigma) \in \bigl\{\, \emptyset,\ \{\tau_p\} \,\bigr\}.
\]
That is: either nothing happened, or exactly the transition that was
rendered happened, exactly once.
\end{definition}

It is worth being explicit about what this does \emph{not} say.
Fidelity makes no claim that $\tau_p$ is a faithful rendering of
$u$; F1 is unbounded by construction and survives untouched. What
fidelity buys is that F1 becomes the \emph{only} way a session can
go wrong. Every other fault in the model is forced to yield
$\emptyset$ rather than an unintended effect, so the residual risk
of the whole pipeline collapses onto a single predicate---``does
$\tau_p$ render $u$''---evaluated over a finite object: four
coordinates and an amount.

This is the sense in which the four-dimensional space earns its
place in the paper. Its role is not descriptive elegance; it is that
it makes ``what was promised'' and ``what was executed'' the same
kind of object, so that comparing them is a decidable finite check
rather than a judgment about language.
\S\ref{sec:fidelity-theorem} proves that
Invariants~\ref{inv:auth}--\ref{inv:idempotent-receiver} suffice,
and \S\ref{sec:fidelity-limits} is candid about the three places the
guarantee stops.

%% file: content/figures/4d_space.tex
\begin{figure*}[t]
\centering
\small
\begin{tabular}{llllll}
\toprule
 & Wallet $w$ & Chain $c$ & Address $a$ & Protocol $p$ & Balance \\
\midrule
Pre-state $s$   & $w_0$ & arbitrum       & $a_w^{\text{EVM}}$ & USDC      & $54{,}331{,}159$ \\
Post-state $s'$ & $w_0$ & hyperliquid-L1 & $a_w^{\text{HL}}$  & perp-USDC & $54{,}331{,}159$ \\
\midrule
\textsc{Deposit} & fixed & \textbf{moved} & \textbf{moved} & \textbf{moved} & conserved \\
\bottomrule
\end{tabular}
\caption{A transition as a coordinate readout: ``deposit all my USDC
to Hyperliquid'' resolved against $\mathcal{S} = W \times C \times A
\times P$. Three of four coordinates move; $W$ is fixed, because the
same mnemonic authorizes both sides. $\mathcal{S}$ is an index set
rather than a geometry, so we render a transition as the pair of
coordinate tuples it relates rather than as a projection. This is
also, deliberately, the shape of the preview the agent must render
before a key is derived (Invariant~\ref{inv:auth},
\S\ref{sec:orchestration}): the reason to
carry all four coordinates explicitly is that it makes ``what was
promised'' a finite object that a user, or a scope check, can compare
against ``what was executed''. Balances are integers in the
protocol's smallest unit; $54{,}331{,}159$ is 54.33 USDC at six
decimals.}
\label{fig:4dspace}
\end{figure*}

%% file: content/04_orchestration.tex
\section{Agent Orchestration}
\label{sec:orchestration}

\input{content/figures/architecture}

\subsection{Mapping natural language to transitions}

The agent receives a user request in natural language and must
commit to a single transition tuple before touching any signing
key. This step is where the system's guarantee is weakest and
where we make the smallest claim: no runtime layer can decide
whether a tuple faithfully renders an intent expressed in
natural language. What the layer can do is force the mapping to
be \emph{explicit, singular, and inspectable} before anything
irreversible happens, so that the residual risk is concentrated
in one bounded comparison rather than spread across the whole
execution. Two mechanisms serve that end.

\emph{Skill-loading} injects only those routing rules whose
triggers match the current message (declared in per-protocol
markdown files with regex frontmatter), so the system prompt
grows $O(1)$ rather than $O(N)$ in the number of supported
protocols. \emph{Preview--confirm} is the mechanism behind
Invariant~\ref{inv:auth} below: the agent must render the tuple
it has chosen and obtain confirmation against that rendering
before the key is derived. A request that does not resolve to a
single tuple therefore cannot be executed at all: the correct
behavior is a clarification request, not a best guess
(Example~\ref{ex:ambiguous}).

\subsection{Deriving the invariants}
\label{sec:deriving}

A list of seven safety rules invites the question of where the
seven came from---whether they follow from anything, or are simply
the bugs we happened to hit written up as principles. We therefore
derive them, and defer the production failure classes to
\S\ref{sec:failure-classes} where they serve as a check on the
derivation rather than as its source.

The derivation is a case analysis on
Definition~\ref{def:fidelity}. A session's realized multiset
$\mathrm{real}(\sigma)$ must equal $\emptyset$ or $\{\tau_p\}$, so
a violation is an element of $\mathrm{real}(\sigma)$ that is
either (a)~not $\tau_p$, (b)~a second copy of $\tau_p$, or
(c)~something derived from $\tau_p$ but not equal to it. Each
branch generates its obligations by asking what a runtime layer
must establish, and \emph{when}, to exclude it.

\medskip\noindent\textbf{(a) Nothing unrendered may execute.}
Executing requires a key. So the key must not be derivable
without a tuple that has been rendered and confirmed---this is
Invariant~\ref{inv:auth}. Distinctly, the fact that \emph{a}
confirmation exists does not establish that \emph{this caller}
was entitled to it, which under F6 is a separate question and
yields Invariant~\ref{inv:delegation}. These two are often
conflated as ``authorization''; they answer different questions
(\emph{was this tuple approved} versus \emph{may you approve it}),
and a system with only one of them is exploitable in a way we
observed in practice.

\medskip\noindent\textbf{(c) Nothing derived-but-different may
execute.} A tuple approved at planning time may be inapplicable at
broadcast time under F2. The layer must therefore re-establish
$s \in \mathrm{dom}(\tau_p)$ immediately before signing, against
ground truth rather than a cached view---Invariant~\ref{inv:preflight}.
The alternative failure here is subtle and worth naming: a system
that silently adapts the transition to the available balance
rather than refusing has produced a transition the user never saw,
which is case~(a) wearing case~(c)'s clothes.

\medskip\noindent\textbf{(b) Nothing may execute twice.}
Multiplicity has three independent sources, and closing one does
nothing for the others. \emph{The executor's own belief:} if the
system can report failure on a transition that in fact landed, it
manufactures the trigger for a retry, so outcomes that cannot be
observed must be reported as uncertain
(Invariant~\ref{inv:postobs}) and outcomes that cannot be
confirmed must not be reported as success
(Invariant~\ref{inv:phantom}). \emph{The planner (F4):} even given
an honest uncertainty signal, nothing stops a planner from calling
again, so re-issue must be gated on positively confirming absence
(Invariant~\ref{inv:plan-retry}). \emph{The channel (F5):}
duplication upstream of the planner is invisible to both of the
above, so the boundary that ingests intent must deduplicate
(Invariant~\ref{inv:idempotent-receiver}).

\medskip
Table~\ref{tab:invariants} records the result. We claim
completeness only with respect to this decomposition, which is a
bounded claim and the only one we can support: given
Definition~\ref{def:fidelity} and faults F1--F6, every way to
violate fidelity falls in branch (a), (b) or~(c), and each branch's
obligations are discharged by the invariants listed against it.
This is not a claim that no other failures exist---F1 sits outside
the analysis entirely, and \S\ref{sec:fidelity-limits} names three
further limits.

\begin{table}[h]
\centering\small
\caption{The seven invariants, by the fidelity-violation branch
that generates each. \emph{Tier} indicates where each is enforced:
\textbf{E} executor-local, \textbf{X} cross-tier, \textbf{P}
plan-level, \textbf{R} HTTP receiver.}
\label{tab:invariants}
\begin{tabular}{clcl}
\toprule
\# & Name & Tier & Excludes \\
\midrule
\ref{inv:auth}                 & Confirm gate                & E & (a) unrendered \\
\ref{inv:preflight}            & Pre-state sufficiency       & E & (c) inapplicable at broadcast \\
\ref{inv:postobs}              & Post-state observability    & E & (b) false failure report \\
\ref{inv:phantom}              & No phantom success          & E & (b) false success report \\
\ref{inv:delegation}           & Scoped delegation           & X & (a) unentitled caller \\
\ref{inv:plan-retry}           & Plan-level retry safety     & P & (b) planner re-issue \\
\ref{inv:idempotent-receiver}  & Idempotent intent receiver  & R & (b) channel redelivery \\
\bottomrule
\end{tabular}
\end{table}

Two structural observations fall out of the derivation that were
not visible when the invariants were merely listed.
Invariants~\ref{inv:auth}--\ref{inv:phantom} are
\emph{executor-local}: they hold within the wallet agent
regardless of who issued the request, and are properties of a
single signing-and-broadcast cycle. The remaining three are not,
and each becomes load-bearing only when a specific structural
feature is present---Invariant~\ref{inv:delegation} once the
executor has more than one caller type,
Invariant~\ref{inv:plan-retry} once a planner can retry without a
human in the loop, Invariant~\ref{inv:idempotent-receiver} once
intent arrives over a network boundary. A single-user,
single-process, human-driven wallet needs only the first four,
which is why systems that stop there appear correct until they
grow one of those three features.

Invariants~\ref{inv:plan-retry} and~\ref{inv:idempotent-receiver}
in particular attack duplication from opposite sides of the same
boundary---sender-side verification versus receiver-side
deduplication---and neither subsumes the other: retry safety
cannot see a duplicate that arrives before the planner runs, and
an idempotent receiver cannot see a second call the planner makes
on its own initiative.

\begin{invariant}[Confirm gate]\label{inv:auth}
No signing key is derived, and therefore no transition is
authorized, until an explicit confirmation has been recorded
against a \emph{rendered} transition tuple. The write tool runs
in two phases: a \texttt{confirm=false} call resolves the user's
request to a single tuple $\tau$ and returns it as a preview
(source and destination coordinates, amount in both integer and
human-decimal form, fees, expected output); only a subsequent
\texttt{confirm=true} call carrying that same tuple releases the
key. In attended operation the confirmation is the human's; in
unattended operation it is admission by
Invariant~\ref{inv:delegation}'s scope check, which is the
weaker of the two and bounds the resulting authority rather than
checking the tuple's content
(\S\ref{sec:discussion}).
\end{invariant}

This gate is what makes the rest of the stack statable: it
forces the system to commit to exactly one point of
$\mathcal{S}$ before any irreversible step, so that ``what was
promised'' is a finite object the remaining invariants can be
stated against. It is a property of the \emph{orchestration}
layer, distinct from the substrate's own signature requirement
(\S\ref{sec:formalism}, Definition~\ref{def:transition}): the
ledger guarantees that some key signed whatever was broadcast,
but not that what was broadcast is what the user was shown.

A gate is only as strong as its cheapest bypass. In our
deployment that bypass was a type coercion rather than a logic
error: MCP serializes \texttt{bool} arguments as the literal
string \texttt{"false"} on some client paths, and Python's
non-empty-string truthiness silently promotes it to
\texttt{True}, admitting an unconfirmed call as a confirmed one
(failure class~3, \S\ref{sec:failure-classes}). Every write tool
now parses string forms explicitly. We report this because it
illustrates a general point about where these invariants
actually fail---at the type boundary between the agent's
serialization and the executor's interpretation, not in the
policy the invariant states.

\begin{invariant}[Pre-state sufficiency]\label{inv:preflight}
Every transition is gated by a pre-broadcast check of the
current on-chain balance. A transition that would set any
coordinate below zero is rejected \emph{before} signing, with a
message naming the token, current balance, and shortfall.
Crucially, the preflight check must bypass eventually-consistent
indexer caches and hit the L1 RPC \texttt{balanceOf} directly; a
cached read would permit indexer-drift failures. The L1 RPC
remains a trust assumption: the current deployment uses a pinned
set of providers per chain with cross-provider sanity checks on
the read. Light-client verification (header sync + inclusion
proof) would close the surface and is flagged as future work
(\S\ref{sec:discussion}).
\end{invariant}

\begin{invariant}[Post-state observability]\label{inv:postobs}
After broadcast, the tool polls the chain's transaction-lookup
RPC for up to 12 seconds before reporting success. If the poll
times out, the tool reports \emph{uncertainty}---not failure---
and includes the block-explorer URL so the agent can convey
that nuance rather than parroting ``failed''.
\end{invariant}

\begin{invariant}[No phantom success]\label{inv:phantom}
A broadcast that returns a hash but never propagates to an
indexer is treated as uncertain; the tool refuses to claim
``Sent!'' on a hash whose existence on chain cannot be verified.
\end{invariant}

\begin{definition}[Authorization predicate]\label{def:auth}
Fix a finite scope alphabet $\mathfrak{S}$. For an actor identity
$\alpha = (\textsf{caller\_type}, \textsf{caller\_id})$, a scope
set $\Sigma \subseteq \mathfrak{S}$, and a transition $\tau$
requiring scope set $\Sigma_\tau \subseteq \mathfrak{S}$, the
\emph{authorization predicate} is $\mathrm{auth}(\alpha, \Sigma,
\tau) = 1$ iff $\Sigma_\tau \subseteq \Sigma$. The predicate is
independent of $\alpha$'s type: a human caller and an avatar
caller with the same $\Sigma$ are indistinguishable to the
executor; the difference is observed only in the audit log.
\end{definition}

\begin{invariant}[Scoped delegation]\label{inv:delegation}
Every write-capable tool call carries an actor identity
$\alpha$ and a scope set $\Sigma$ derived from the auth context.
The executor admits the call iff $\mathrm{auth}(\alpha, \Sigma,
\tau) = 1$, and the audit log records both $\alpha$ and
$\Sigma_\tau$ alongside the transition tuple. For human callers
$\Sigma$ is implicit (the session token grants every scope in
$\mathfrak{S}$); for non-human callers $\Sigma$ is explicit and
bounded by a delegation token, so an avatar holding a $\Sigma =
\{\texttt{transfer}\}$ token cannot invoke a \texttt{swap} even
under prompt injection. The invariant decouples authentication
(``who is calling'') from authorization (``what may they do'').
\end{invariant}

\begin{invariant}[Plan-level retry safety]\label{inv:plan-retry}
A plan $\pi = \tau_n \circ \cdots \circ \tau_1$ issued by the
avatar tier may not re-issue any irreversible transition
$\tau_i$ on a reported failure or uncertainty until the avatar
has positively confirmed, by querying ground-truth on-chain
state through the same indexer-bypassing path required by
Invariant~\ref{inv:preflight}, that $\tau_i$'s effect is absent.
\emph{Verify-then-retry} is required; bare retry-on-failure
(``the call returned an error, so I'll call it again'') is
prohibited. The invariant is enforced at the avatar tier
because Invariant~\ref{inv:postobs}'s ``uncertainty'' signal
lives inside one executor-local cycle and does not by itself
prevent the avatar from re-invoking the same tool a few seconds
later.
\end{invariant}

\begin{invariant}[Idempotent intent receiver]\label{inv:idempotent-receiver}
For any externally-originated intent $i$ delivered over an
at-least-once channel (webhook, message queue, retried HTTP),
the receiver hashes a canonical pair $h(i) =
H(\textsf{sender\_id}, \textsf{request\_bytes})$ and replays the
recorded response on every duplicate delivery within a fixed
deduplication window $\Delta$. The invariant decouples
\emph{network duplication} (at-least-once delivery semantics
every public HTTP boundary inherits) from \emph{intent
duplication} (the user wanting the same action twice), and
ensures a single user intent produces a single transition tuple
regardless of how many times the carrier retries. Invariants
\ref{inv:plan-retry} and \ref{inv:idempotent-receiver} are dual:
the former handles ambiguous outcomes, the latter ambiguous
deliveries.
\end{invariant}

\subsection{Enforcement}
\label{sec:enforcement}

The invariants above are stated as properties. This section says
how each is actually checked, because a property nobody enforces
is a wish, and because the difference between the two is where our
own deployment failed (\S\ref{sec:failure-classes}).

Figure~\ref{fig:writepath} gives the write path for a single
transition. The ordering is not incidental: the balance read must
happen after the tuple is resolved (otherwise there is nothing to
check against) and before the preview is rendered (otherwise the
user confirms a transition already known to be inapplicable), and
the key must be derived after confirmation and released before
nothing else.

\begin{figure}[t]
\small
\begin{verbatim}
write_tool(request, confirm):
  tau  = resolve(request)          # may fail -> clarify
  if not authorized(actor, scope, tau):
      return DENIED                # I5
  bal  = balance_at_L1(tau.source) # I2: no cache
  if bal < tau.amount:
      return INSUFFICIENT(bal)     # I2
  if not is_true(confirm):         # I1: parse "false"
      return PREVIEW(tau, bal, fee)
  key  = derive(tau.wallet)        # only past this point
  txid = broadcast(sign(tau, key))
  for _ in range(POLL_WINDOW):     # I3: <= 12s
      if lookup(txid): return SENT(txid)
  return UNCERTAIN(txid, explorer_url)   # I3, I4
\end{verbatim}
\caption{The executor-local write path. Every early return is a
refusal that leaves $\mathrm{real}(\sigma) = \emptyset$; the only
path to a signature runs through all four executor-local checks.
Note that the terminal state on an unconfirmable broadcast is
\texttt{UNCERTAIN}, never \texttt{FAILED}---this is what denies a
retrying planner its trigger (\S\ref{sec:fidelity-theorem}).}
\label{fig:writepath}
\end{figure}

Table~\ref{tab:enforcement} locates each invariant: the process it
runs in, what it does on violation, and the evidence it leaves
behind. The last column matters for a reason worth stating
separately---an invariant that refuses silently is
indistinguishable in production from one that was never reached,
so each check writes a row naming itself. This is what makes the
false-positive analysis in \S\ref{sec:evaluation} possible at all.

\begin{table}[h]
\centering\small
\caption{Where each invariant executes, what it does on violation,
and what it leaves in the audit log.}
\label{tab:enforcement}
\begin{tabular}{cp{0.21\linewidth}p{0.25\linewidth}p{0.28\linewidth}}
\toprule
\# & Runs in & On violation & Evidence \\
\midrule
\ref{inv:auth} & write tool, per call &
  return preview, no key derivation & preview row, \texttt{confirmed} flag \\
\ref{inv:preflight} & write tool, pre-sign &
  refuse, name shortfall & refusal row with observed balance \\
\ref{inv:postobs} & write tool, post-broadcast &
  report \texttt{UNCERTAIN} & outcome column, poll duration \\
\ref{inv:phantom} & write tool, post-broadcast &
  withhold success claim & outcome column, txid \\
\ref{inv:delegation} & auth middleware, pre-dispatch &
  reject call & \texttt{caller\_type}, $\Sigma_\tau$ \\
\ref{inv:plan-retry} & avatar loop, per plan step &
  block re-issue pending verification & run outcome \textsf{escalated} \\
\ref{inv:idempotent-receiver} & HTTP receiver, pre-handler &
  replay recorded response & dedup-hit counter \\
\bottomrule
\end{tabular}
\end{table}

Two of the seven are not executor-local and are worth spelling out,
since a reader cannot infer their mechanism from the write path.

\medskip\noindent\textbf{Scoped delegation
(Invariant~\ref{inv:delegation}).} The check lives in the auth
middleware, ahead of dispatch, so it applies uniformly to every
write-capable tool rather than being re-implemented per tool. A
delegation token carries the triple $(\textsf{user},
\textsf{avatar}, \Sigma)$ as signed claims; the middleware resolves
the requested tool to its required scope set $\Sigma_\tau$ and
evaluates Definition~\ref{def:auth}'s predicate. Because the
predicate is independent of caller type, a human session token and
an avatar delegation token traverse identical code---the avatar
path is not a privileged bypass, which is precisely what makes it
auditable. Placement is the design decision: a per-tool check would
be seven opportunities to forget.

\medskip\noindent\textbf{Idempotent receiver
(Invariant~\ref{inv:idempotent-receiver}).} The check lives at the
HTTP boundary, ahead of any handler, and keys a sliding-window
cache on $H(\textsf{sender\_id}, \textsf{request\_bytes})$ with
$\Delta = 30$\,s. On a hit the recorded response is replayed
verbatim, so the caller cannot distinguish a deduplicated retry
from the original---which is the point, since a caller that could
tell would retry differently. Hashing the raw bytes rather than a
parsed payload is deliberate: it makes the check independent of
schema evolution, at the cost of missing semantically identical
requests that differ in serialization. For an at-least-once
carrier replaying its own buffered bytes, that trade is the right
one.

\medskip
What none of this gives us is assurance that the implementations
match the properties. Failure class~3 is exactly that gap: the
policy in Invariant~\ref{inv:auth} was right and
\texttt{is\_true(confirm)} in Figure~\ref{fig:writepath} was, for a
period, wrong. The enforcement points are where to audit, not
evidence that the audit passed.

\subsection{Multi-tier orchestration: avatars driving wallets}
\label{sec:multi-tier}

A separate \emph{avatar service} runs persistent agents---one per
user-installed template---that act as the user when the user is
not actively signed in. By design, avatars cannot reach any
external system directly: every action they take is a call into
an \emph{executor agent} (the wallet for fund-moving operations;
an analyst service for read-only news, prediction-market, and
sentiment data) over the executor's normal user-facing API. The
threat-model implications were established in
\S\ref{sec:background}.

The orchestration architecture is two-tiered. The avatar-side
tier decides \emph{which executor to call and with what intent},
runs on a trigger (scheduled cadence, external webhook, or
predicate over indexed state), and may interleave calls to
multiple executors within one behavior. The executor-local tier
is unchanged from the single-actor description above: it maps
the intent to a transition tuple in $\mathcal{S}$ and applies
Invariants~\ref{inv:auth}--\ref{inv:phantom}.

The cross-tier security model is
Invariant~\ref{inv:delegation}. The avatar holds a long-lived
\emph{delegation token} bound to a specific $(\textsf{user},
\textsf{avatar}, \Sigma)$ triple, replacing the human's session
token in the auth header but otherwise traversing the same code
path. The wallet does not need any privileged backdoor for
avatars; the token is just a different kind of credential the
same auth middleware recognizes. Conversation threads are tagged
with $(\textsf{user\_id}, \textsf{owner\_type},
\textsf{owner\_id})$ so an avatar's working memory does not
pollute the human user's chat thread, and the audit log
distinguishes human-initiated from avatar-initiated rows by the
\texttt{caller\_type} column. Avatar-initiated rows are by
construction not adversarial (the avatar is software the user
installed and parameterized, not an attacker) and form the
cleanest source for the false-positive rate of each invariant
under realistic load (\S\ref{sec:evaluation}).

\subsection{Failure modes observed in the wild}
\label{sec:failure-classes}

\input{content/figures/transition_flow}

We recorded nine classes of safety bug while building and
iterating on the system. Chronologically these came first: the
system was built, it failed in these ways, and the invariants were
written down afterwards. We present them second, and as a
\emph{check} on \S\ref{sec:deriving} rather than as its source,
because the interesting question is not whether we can name the
bugs we fixed but whether the derivation predicts them. It does:
every class below lands in a branch the case analysis already
identified, and no class required inventing an eighth invariant.
Table~\ref{tab:failure-classes-summary} maps each to the invariant
that catches it; \S\ref{sec:evaluation} reports per-class incidence
in the deployment record
(Table~\ref{tab:failure-classes}).

The exception proves the point. Class~5 (provider gap) is the one
entry with no invariant against it, and it is not an oversight:
failing to find a route is a \emph{liveness} failure, not a
fidelity failure. Nothing executes, so
Definition~\ref{def:fidelity} is satisfied trivially and the
derivation has nothing to say. It appears in the table because it
is a real operational problem, not because the safety analysis
covers it.

\begin{table}[h]
\centering\small
\caption{Failure classes observed in production, by the
invariant that catches each. Class~5 is a liveness failure and
falls outside the fidelity analysis; see text.}
\label{tab:failure-classes-summary}
\begin{tabular}{cp{0.72\linewidth}c}
\toprule
\# & Class (one-line description) & Inv. \\
\midrule
1 & Phantom success: RPC returns hash for a tx that never propagates &
    \ref{inv:phantom} \\
2 & Phantom failure: legitimate tx lands but verify window too short &
    \ref{inv:postobs} \\
3 & Preview-gate bypass via \texttt{confirm="false"} bool coercion &
    \ref{inv:auth} \\
4 & Wrong-wallet selection: source has no balance &
    \ref{inv:preflight} \\
5 & Provider gap: no single bridge covers the route &
    liveness \\
6 & Slow indexer drift: cached \texttt{balanceOf} lags chain state &
    \ref{inv:preflight} \\
7 & Stale agent context: re-emits stale negative conclusion on
    repeat request & \ref{inv:plan-retry} \\
8 & Documented-atomicity drift: ``atomic'' exchange API applies
    half on partial failure & \ref{inv:preflight} \\
9 & Duplicate webhook delivery: at-least-once retry processes intent
    twice & \ref{inv:idempotent-receiver} \\
\bottomrule
\end{tabular}
\end{table}

The two LLM-layer classes (7 stale context, 8 documented-atomicity
drift) are the ones whose defense is least obvious from the
single-cycle invariants \ref{inv:auth}--\ref{inv:phantom}.
Stale-context resolution lives at the avatar tier: skill-level
rules force a fresh state query on repeat write operations
before consulting the agent's accumulated context, so an earlier
turn's negative conclusion cannot block a later turn's
re-evaluation. Documented-atomicity drift is a client-side
preflight (\ref{inv:preflight}): even when an exchange documents
an action as atomic, the client checks the precondition the
exchange claims to enforce, so a server-side validation failure
on the replacement half of a cancel-and-replace cannot leave
the user with no resting order.

\subsection{The fidelity theorem}
\label{sec:fidelity-theorem}

\S\ref{sec:deriving} argued informally that the seven invariants
close every branch of the case analysis. We now discharge that
obligation formally against the target set in \S\ref{sec:fidelity}.
The proof follows the same three branches, which is the point: the
derivation and the theorem are the same argument stated at
different levels of rigour, not two independent claims that happen
to agree.

\begin{theorem}[Execution fidelity]\label{thm:fidelity}
Let $\sigma$ be a session executed under
Invariants~\ref{inv:auth}--\ref{inv:idempotent-receiver}, subject to
faults drawn from \textnormal{F1--F6} and the assumptions of
\S\ref{sec:fidelity}. Then $\sigma$ satisfies execution fidelity:
$\mathrm{real}(\sigma) \in \{\emptyset, \{\tau_p\}\}$.
\end{theorem}

\begin{proof}
Since $\mathrm{real}(\sigma)$ is a multiset of transitions, a
violation takes one of exactly three forms: it contains a
transition other than $\tau_p$; it contains $\tau_p$ with
multiplicity greater than one; or---the degenerate reading of the
first---it contains a transition derived from $\tau_p$ but not
equal to it. We rule out each.

\emph{(i) No transition other than $\tau_p$.} By the signature
condition of Definition~\ref{def:transition}, any transition that
decreases a balance carries a signature from a key derivable from
$w$, so every element of $\mathrm{real}(\sigma)$ is preceded by a
key derivation inside the executor. Invariant~\ref{inv:auth} makes
key derivation conditional on a confirmation recorded against a
rendered tuple, and a session renders exactly one such tuple
(Definition~\ref{def:session}). Hence every element of
$\mathrm{real}(\sigma)$ equals $\tau_p$. Note that F1 is untouched
by this step: $\tau_p$ may be a poor rendering of $u$, and the
claim is only that nothing outside the rendered tuple executes.
Under F6 the confirmation is a scope admission rather than a human
act; Invariant~\ref{inv:delegation} rejects a caller whose scope
set does not cover $\tau_p$, yielding $\mathrm{real}(\sigma) =
\emptyset$.

\emph{(ii) No duplicate.} A second copy of $\tau_p$ can arise on
three paths. \emph{Within one executor cycle:} a retry would have
to be triggered by a definite failure report, but under F3 an
unobservable outcome is reported as uncertainty rather than failure
(Invariant~\ref{inv:postobs}), and a hash that cannot be confirmed
is never reported as success (Invariant~\ref{inv:phantom}); the
trigger is therefore not produced. \emph{From the planner (F4):}
Invariant~\ref{inv:plan-retry} forbids re-issuing $\tau_p$ until
its absence has been positively confirmed against ground-truth
state, so a re-issue following a landed first attempt is refused.
\emph{From the channel (F5):} Invariant~\ref{inv:idempotent-receiver}
replays the recorded response for a duplicate delivery within the
window $\Delta$, so the duplicate produces no second session and
hence no second transition.

\emph{(iii) Nothing derived-but-different.} Under F2 the pre-state
may have drifted so that $s \notin \mathrm{dom}(\tau_p)$ at
broadcast time. Invariant~\ref{inv:preflight} rejects the
transition in that case rather than adjusting it to fit the
available balance, so the outcome is $\emptyset$ and not a
transition the user never saw. This step is why preflight must read
ground truth rather than a cached indexer value: a stale read can
report $s \in \mathrm{dom}(\tau_p)$ when it is false, which admits
exactly the derived-but-different case the invariant exists to
exclude.

Every fault in F1--F6 is therefore either excluded or forced to
$\emptyset$, except F1, which by construction acts only on the
choice of $\tau_p$ itself.
\end{proof}

\begin{corollary}[Residual-risk localization]\label{cor:localization}
Under Invariants~\ref{inv:auth}--\ref{inv:idempotent-receiver}, the
only fault that can produce an unintended irreversible effect is
F1. Every other fault class yields $\emptyset$.
\end{corollary}

Corollary~\ref{cor:localization} is the paper's actual claim, and
it is deliberately modest. We do not solve the intent-mapping
problem; we confine it. The engineering value is that an unbounded
semantic question---did the agent understand the user---has been
reduced to a bounded one: does this four-coordinate tuple and
amount match what was asked. That question is answerable by a human
in a fraction of a second, which is what makes the confirm gate a
usable control rather than a formality.

\subsection{Where the guarantee stops}
\label{sec:fidelity-limits}

Three limits, stated plainly because each is load-bearing.

\medskip\noindent\textbf{The unattended case is materially weaker.}
In attended operation a human evaluates ``does $\tau_p$ render
$u$'' at the confirm gate. In unattended operation nobody does. The
bound degrades from a check on the tuple's \emph{content} to
Invariant~\ref{inv:delegation}'s check on the caller's
\emph{authority}, and those are not the same thing: an avatar
holding $\Sigma = \{\texttt{transfer}\}$ that resolves ``send one
unit'' into a transfer of ten is inside its scope and executes.
Fidelity holds---the ten-unit transfer was previewed, and it
happened once---while the user's intent is still violated. Scope
sets today name tools; extending them to carry predicates over the
transition tuple (per-period caps, destination allowlists, rate
limits) would let $\Sigma$ bound content as well as authority and
is the most valuable single extension we are aware of
(\S\ref{sec:discussion}).

\medskip\noindent\textbf{Duplicate suppression is
$\Delta$-bounded.} Invariant~\ref{inv:idempotent-receiver} holds
within its deduplication window. A redelivery arriving after
$\Delta$ opens a \emph{new} session, and the theorem is stated per
session, so it says nothing about the pair. Choosing $\Delta$ is
therefore a safety parameter and not a cache-tuning decision, a
point the Stripe-style idempotency-key literature makes for
payments and which transfers directly.

\medskip\noindent\textbf{Plans get prefix consistency, not
atomicity.} For a plan $\pi = \tau_n \circ \cdots \circ \tau_1$,
applying the theorem stepwise gives that $\mathrm{real}(\pi)$ is a
prefix of $(\tau_1, \ldots, \tau_n)$ with each element previewed
and no duplicates. That is strictly weaker than atomicity: a plan
interrupted after $\tau_i$ leaves a reachable intermediate state
that no one asked for---a swap completed but its follow-on bridge
not. Fidelity is a per-transition property and cannot see this;
lifting it to a transactionality predicate over compound
transitions is the open problem we regard as most worth attacking
(\S\ref{sec:discussion}).

\medskip
Finally, a note on how to read the theorem in practice. It is
conditional on the invariants holding, so its useful form is the
contrapositive: every failure class in
Table~\ref{tab:failure-classes-summary} is an instance of some
invariant not in fact holding in deployed code. Failure class~3 is
the sharpest example---the policy stated by
Invariant~\ref{inv:auth} was correct, and a string-to-boolean
coercion made the implementation not satisfy it. The theorem tells
you where to look; it does not certify the code.

%% file: content/figures/architecture.tex
\begin{figure*}[t]
\centering
\begin{tikzpicture}[
  box/.style={rectangle, rounded corners=2pt, draw=black!70, thick,
              minimum width=3.0cm, minimum height=0.85cm,
              align=center, font=\small\sffamily},
  agent/.style={box, fill=blue!8, draw=blue!60!black},
  tool/.style={box, fill=green!10!white, draw=green!50!black},
  sign/.style={box, fill=red!10, draw=red!55!black},
  external/.style={box, fill=gray!15, draw=gray!60},
  pin/.style={font=\scriptsize, text=red!55!black, anchor=west,
              align=left, inner sep=1pt},
  tie/.style={gray!55, thin},
  flow/.style={-{Stealth[length=2.4mm]}, thick, black!75},
  edgelbl/.style={font=\scriptsize, text=black!60, align=center},
]
\node[agent]    (user)    at (0, 0.0)  {User request\\\scriptsize natural language};
\node[agent]    (agent)   at (0,-1.5)  {Agent\\\scriptsize LLM + skill routing};
\node[tool]     (tool)    at (0,-3.0)  {Write tool\\\scriptsize preview / confirm};
\node[sign]     (sign)    at (0,-4.5)  {Signer\\\scriptsize key derivation};
\node[external] (api)     at (0,-6.0)  {RPC / exchange API};
\node[external] (chain)   at (0,-7.5)  {Public ledger\\\scriptsize EVM / BTC / HL-L1};

\draw[flow] (user)  -- (agent);
\draw[flow] (agent) -- (tool);
\draw[flow] (tool)  -- (sign);
\draw[flow] (sign)  -- (api);
\draw[flow] (api)   -- (chain);

\draw[flow, dashed] (chain.west) -- ++(-1.4,0)
  |- node[edgelbl, pos=0.25, anchor=east, xshift=-1mm] {verify\\poll} (tool.west);

\node[pin] (p1) at (2.0,-2.75) {\textbf{I1}~confirm gate};
\node[pin] (p2) at (2.0,-3.25) {\textbf{I2}~pre-state sufficiency};
\node[pin] (p3) at (2.0,-7.25) {\textbf{I3}~post-state observability};
\node[pin] (p4) at (2.0,-7.75) {\textbf{I4}~no phantom success};

\draw[tie] (tool.east)  -- (1.9,-3.0);
\draw[tie] (chain.east) -- (1.9,-7.5);
\end{tikzpicture}
\caption{Orchestration pipeline in \sysname{}. A natural-language
request flows top to bottom: the agent resolves it to a single
transition tuple, the write tool renders that tuple and holds it until
confirmation (\textbf{I1}) and a ground-truth balance check
(\textbf{I2}) both pass, and only then is a key derived. The dashed
return path is the post-broadcast verify poll, which supplies
\textbf{I3} and \textbf{I4} before the tool reports any outcome back
to the agent. Invariants are stated in
\S\ref{sec:orchestration}; the pins here mark where each one executes.}
\label{fig:architecture}
\end{figure*}
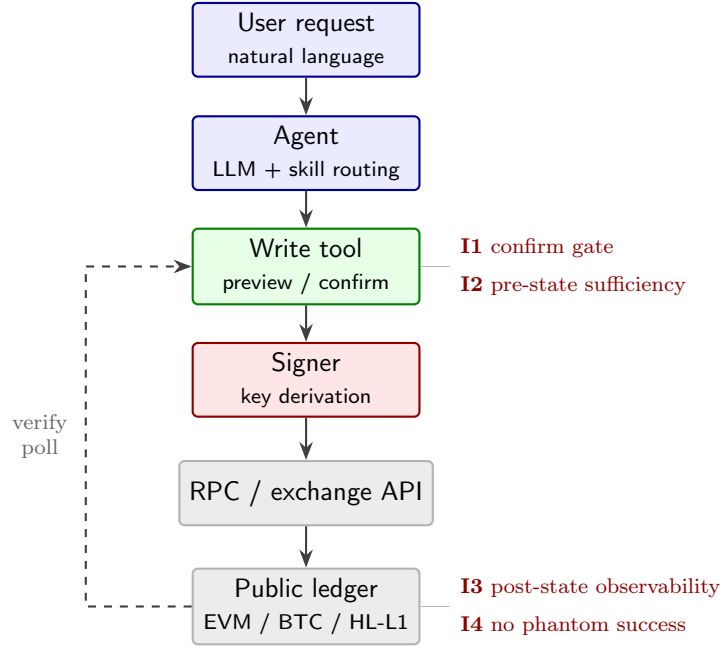

%% file: content/figures/transition_flow.tex
\begin{figure*}[t]
\centering
\begin{tikzpicture}[
  state/.style={rectangle, rounded corners=2pt, draw=black!70, thick,
                minimum width=2.4cm, minimum height=0.75cm,
                align=center, font=\small\sffamily, fill=gray!8},
  start/.style={state, fill=blue!10, draw=blue!55!black},
  final/.style={state, fill=green!12, draw=green!55!black},
  halt/.style={rectangle, rounded corners=2pt, draw=red!50!black,
               dashed, thick, minimum width=3.4cm, minimum height=0.7cm,
               align=center, font=\scriptsize\sffamily, fill=red!4},
  arrow/.style={-{Stealth[length=2.2mm]}, thick, black!75},
  fail/.style={-{Stealth[length=2.2mm]}, dashed, red!60!black, thick},
  breach/.style={-{Stealth[length=2.2mm]}, dashed, red!70!black, thick},
  faillbl/.style={font=\scriptsize, text=red!55!black, align=center},
]
\node[start] (intent)  at (0, 0.0)  {Intent};
\node[state] (plan)    at (0,-1.15) {Plan};
\node[state] (preview) at (0,-2.30) {Preview};
\node[state] (confirm) at (0,-3.45) {Confirm};
\node[state] (sign)    at (0,-4.60) {Sign};
\node[state] (bcast)   at (0,-5.75) {Broadcast};
\node[final] (verify)  at (0,-6.90) {Verify};

\foreach \a/\b in {intent/plan, plan/preview, preview/confirm,
                   confirm/sign, sign/bcast, bcast/verify}
  \draw[arrow] (\a) -- (\b);

\node[halt] (h2) at (4.6,-3.45) {refuse: preview stale\\(class 2)};
\node[halt] (h4) at (4.6,-4.60) {refuse: insufficient balance\\(classes 4, 8)};
\node[halt] (h1) at (4.6,-6.90) {report \emph{uncertain}, not failed\\(classes 1, 2)};

\draw[fail] (confirm.east) -- (h2.west);
\draw[fail] (sign.east)    -- (h4.west);
\draw[fail] (verify.east)  -- (h1.west);

\draw[breach] (plan.west) -- ++(-1.0,0) |- (sign.west);
\node[font=\small, text=red!70!black] at (-2.2,-2.875) {$\times$};
\node[faillbl, anchor=east] at (-2.45,-2.875)
  {gate bypassed\\(class 3)};
\end{tikzpicture}
\caption{State-machine view of a single transition. Solid arrows are
the happy path. Dashed arrows on the right are \emph{guarded exits}:
outcomes a safety invariant produces deliberately, annotated with the
failure class from \S\ref{sec:failure-classes} that motivated the
guard. The dashed arrow on the left is different in kind---it is the
one edge that must not exist, the boolean-marshalling coercion
(class~3) that let an unconfirmed call reach the signer. The
distinction matters for reading the ablations in
\S\ref{sec:evaluation}: disabling a guard removes a right-hand exit,
while a bypass re-opens the left-hand edge.}
\label{fig:transition-flow}
\end{figure*}
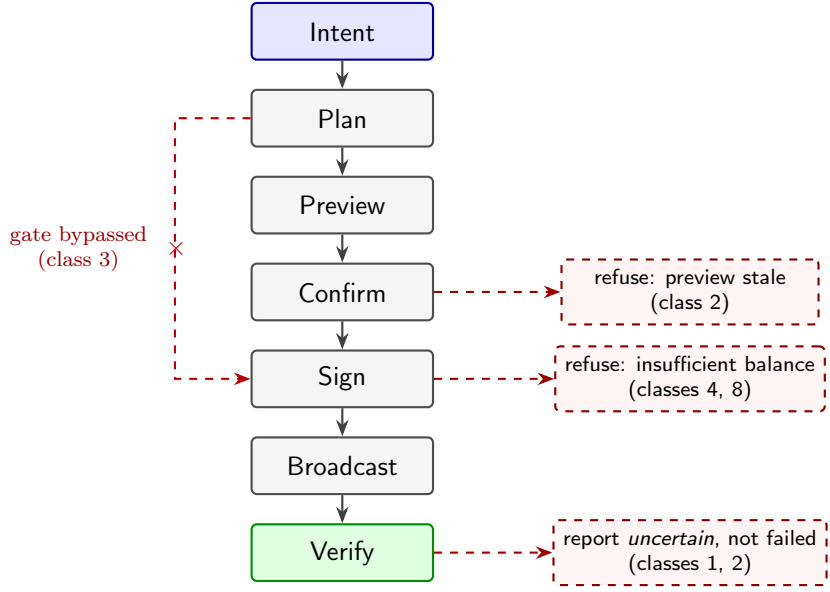

%% file: content/05_evaluation.tex
\section{Empirical Evaluation}
\label{sec:evaluation}

We report two kinds of evidence, and it is worth separating them
clearly because they answer different questions and carry very
different weight.

The \emph{controlled} evidence (\S\ref{sec:controlled}) is the
evaluation proper: an adversarial-prompt suite targeting each
failure class, run across four backing models, with each invariant
ablated in turn. It is designed to answer whether the invariant
stack changes agent behavior, and it supports a comparison because
conditions differ by one variable at a time.

The \emph{deployment record} (\S\ref{sec:deployment}) is not an
experiment and we do not present it as one. It is 108 write
operations the system performed on real wallets, which establishes
that the system runs in production rather than in a simulator, and
it is the source of the failure taxonomy in
\S\ref{sec:failure-classes}. There is no control group and no
counterfactual; a reader should take from it that the failure
classes are real and were observed, not that the invariants are
thereby validated. Presenting a track record as though it were a
result is a common enough move in systems papers that we prefer to
disclaim it explicitly.

\subsection{Controlled experiments}
\label{sec:controlled}
\label{sec:methodology}

\paragraph{Adversarial prompt suite.}
A labeled prompt dataset (\texttt{data/adversarial\_prompts.csv},
$N{=}60$, seven to eight prompts per failure class) targets each
class of \S\ref{sec:orchestration}; class 9 (duplicate webhook) is
exercised by a separate HTTP-replay harness. Each row carries a
target class, an expected outcome label, the target tool, and a
ground-truth rationale. We run the dataset against four backing
models for the wallet's planning loop: Claude, MiniMax-M2.7-highspeed,
Gemini-2.5-flash-lite, and OpenAI~gpt-5-nano, all routed through
the wallet's per-tenant LiteLLM gateway. The Gemini path required
two adapter-level fixes (MCP namespace-prefix stripping and a
validation-hint augmentation for tool-result errors); pre-fix
Gemini runs are not reported.

\paragraph{Invariant ablation.}
The same suite is replayed with one invariant disabled at a time
via a request-scoped header gated by an eval-only env var, so
every condition runs against identical code modulo one boolean.
Tool-level ablations (I\textsubscript{2}, I\textsubscript{4}) are
toggled on the broadcasting executor against the operator's wallet.
The five code-gate invariants (I\textsubscript{1},
I\textsubscript{3}, I\textsubscript{5}, I\textsubscript{6},
I\textsubscript{7}) are exercised by deterministic in-process
drivers that mirror the production code paths (the broadcasting
executor would sign every preview-class prompt on first contact;
the drivers separate the gate's logic from the broadcast side
effect). Two caveats on that choice, which we state here rather
than deferring to the driver sources: the in-process drivers were
authored alongside the gates they exercise, so they test the
gate's logic against inputs chosen by the same author rather than
against production traffic; and a trace-replay version that
toggles each gate inside the broadcasting executor over the
recorded audit log would be the stronger design, which we flag as
follow-up (\S\ref{sec:discussion}).

\paragraph{Naive-ReAct baseline.}
The baseline~\cite{yao2023react} runs the same prompt suite against
a bare ReAct loop with the same wallet tools, but with
preview/confirm, pre-state preflight, and post-broadcast
verification removed at the system-prompt level. Each baseline row
uses the same backing model as its stacked comparison row and the
same tool surface, so a delta is attributable to the orchestration
layer rather than to the model or the tool code. We do not attempt
to reproduce closed-source third-party Web3 agents: reproducing a
full commercial stack inside a comparable sandbox is out of scope,
and a partial reproduction would be more misleading than
informative.

One property of this baseline materially limits what it can show,
and it is easy to miss. The baseline is \emph{single-turn}: the
runner presents the prompt and the tool surface and observes
whether the model emits a \texttt{tool\_use} block for a write
tool. It does not execute the tool, return a tool result, or
continue the agent loop. The resulting number is therefore the
model's \emph{stated intent} under naive scaffolding, not a
measurement of what would happen in a production loop where the
call fires, a result comes back, and the model gets another chance
to reconsider. A high baseline pass rate does not imply the model
is safe in deployment. We use the baseline strictly as a control
for attribution---the finding is the \emph{delta} between a
stacked row and its same-model baseline, never either row read on
its own. Capability-uncertainty refusals (the model declines
because it cannot complete the call schema) are classified
separately from safety-driven refusals by the per-prompt judge, so
the baseline reflects safety reflex rather than incidental schema
confusion.

\paragraph{Metrics and scoring.}
We report pass rate: the fraction of prompts on which the agent
produced the expected-outcome label, with 95\% bootstrap intervals
from 1000 resamples of the per-row Bernoulli trials. Those
intervals are \emph{unpaired}, which deserves a note because it is
the wrong choice for half of what we do with them. The stack and
its baseline see the same 60 prompts, so a paired test---McNemar
over per-prompt outcomes---would be strictly more powerful for a
with-vs-without delta on a fixed backing model, and reporting the
unpaired interval there costs us statistical power we could have
had. We report unpaired for uniformity with the cross-model rows,
where the prompts are shared but the model differs and the paired
design is less clean, and flag the paired test as follow-up
(\S\ref{sec:discussion}). The reader should treat the intervals as
conservative for within-model comparisons and appropriate for
across-model ones. False-positive rates on non-adversarial prompts
and per-invariant $p$-values are out of scope at $N{=}60$.

Each response is classified by a two-stage scorer: a
Gemini-2.5-flash-lite
LLM-as-judge at temperature 0 with schema-constrained output
proposes an outcome category and a satisfies / fails verdict; a
deterministic post-hoc override then applies a safety-equivalence
rule (any conservative refusal, preview, or clarification request
counts as a pass when the expected outcome is non-\texttt{succeed},
since refusing on an under-specified prompt is itself correct
safety behavior). The override is reproducible to the byte given
the judge's classifications; the planning path remains
non-deterministic by construction (live hosted-API calls). A
spot-check re-run on Claude moved the headline by under 3 pp with
per-prompt agreement above 90\%.

\input{content/figures/experiment_forest_claude}
\input{content/figures/experiment_forest_minimax}
\input{content/figures/experiment_forest_gemini}
\input{content/figures/experiment_forest_openai_nano}

Figures~\ref{fig:experiments-claude}--\ref{fig:experiments-openai-nano}
report the per-row pass rate and 95\% bootstrap interval for each
backing-model column; the dashed reference line is the same
model's full-stack adversarial pass rate so each ablation drop
reads against a common ceiling.

\paragraph{Results: write-aggressive pair.}
On Claude the full invariant stack catches \emph{81.7\%} of the
$N{=}60$ suite; the naive-ReAct baseline catches \emph{10.2\%}, a
\emph{71.5}-pp safety delta attributable to the orchestration layer.
On MiniMax-M2.7-highspeed the same comparison reads \emph{86.7\%}
vs \emph{10.0\%} (\emph{76.7}-pp gap). The cross-model agreement on
direction and magnitude is the empirical claim: the safety
contribution replicates across two different write-aggressive
planning models, not just the one the system was tuned on.
Bootstrap intervals overlap in the $78$--$92\%$ range (full stack)
and the $3$--$19\%$ range (baseline), non-overlapping by a wide
margin.

\paragraph{Results: write-cautious and methodology-bounded.}
Gemini-2.5-flash-lite reaches \emph{60.0\%} with the stack and
\emph{56.7\%} on the naive baseline---a $\sim$3-pp gap, because the
model declines write tools on $46/60$ baseline prompts on its own.
We read this as evidence that the stack's measurable safety
contribution depends on the underlying model's
write-aggressiveness: on a model that hesitates natively, the
stack overlaps the model's behavior on this single-turn suite,
though it remains necessary for multi-turn cases
(Invariants~\ref{inv:postobs}, \ref{inv:phantom},
\ref{inv:delegation}, \ref{inv:plan-retry}) that single-turn
refusal cannot guard against.
OpenAI~gpt-5-nano reproduces the write-aggressive baseline pattern
($41/60$ issue a write, point $31.7\%$ at $N{=}60$) but its
safety-stack cells were collected at smaller samples ($N{=}19$--$22$)
because per-turn latency on the 30-tool MCP surface routinely
exceeded the 120-second polling budget; sessions log $150$--$350$
seconds and $185\text{k}$--$372\text{k}$ accumulated input tokens
before terminating on the wallet side. Within the achievable
sample (\emph{36.8\%}, \emph{31.8\%}, \emph{40.0\%} for
adv / I\textsubscript{2}-off / I\textsubscript{4}-off) the lift
overlaps the baseline interval; conditional on the $\sim$40\% of
sessions that converged the agent passed $7/9$, $7/8$, $8/10$
respectively---consistent with the stack working when gpt-5-nano
completes a session. We report the column as a
methodology-bounded measurement.

\paragraph{Per-invariant ablation.}
Tool-level I\textsubscript{2} (balance preflight) and
I\textsubscript{4} (no phantom success) ablations move pass rate
by under 4 pp on the three full-$N$ models: Claude
\emph{83.3\%}/\emph{83.3\%} (vs \emph{81.7\%}), MiniMax
\emph{86.7\%}/\emph{85.0\%} (vs \emph{86.7\%}), Gemini
\emph{61.7\%}/\emph{56.7\%} (vs \emph{60.0\%}). Two reasons:
(i)~the agent also reads balances through the read-only
\texttt{get\_all\_balances} tool and reasons about insufficiency
in its own context, blunting tool-level
I\textsubscript{2} removal; (ii)~under the safety-equivalence
scorer (any conservative refusal counts), both tool-level
ablations sit below the measurement floor at $N{=}60$. The
cleanly discriminating ablations are the code-gate rows:
in-process drivers for I\textsubscript{1}, I\textsubscript{3},
I\textsubscript{5} collapse from $20/20$ control to $5$--$8/20$
ablated, and the cross-tier gates (I\textsubscript{6},
I\textsubscript{7}) collapse to $0/20$. These rows establish that
each gate's code path is load-bearing on the failure class it
targets; the tool-level rows remain in the figure as evidence that
agent-level reasoning absorbs single-tool removal, not as a
per-invariant contribution estimator.

\paragraph{Reproducibility.}
The prompt dataset, every result CSV, the per-label scorer, and
the ablation drivers are versioned alongside the paper source but
are not released with it, for the wallet-privacy reason given
above. We therefore make the weaker claim the artifact supports:
the methodology is specified here in enough detail to be
reimplemented, and the scoring is deterministic given a fixed
judge output, but the numbers are not at present independently
re-derivable from published material. A privacy-preserving release
---the controlled-experiment material, which touches no production
wallet, separated from the on-chain extract---is the obvious
remedy and is not yet done. The deterministic override plus
frozen-judge output makes re-scoring a frozen CSV
byte-reproducible. The hosted
planning path is non-deterministic by construction; we pin the
model identifier and the wallet commit SHA in the result-CSV
commit messages. A full per-seed variance study is flagged as
future work (\S\ref{sec:discussion}).

\subsection{Deployment record}
\label{sec:deployment}

\paragraph{Scope and privacy.}
We report every write-tool invocation logged since the system's
first production use on 2026-04-07. Failure-mode incidents
(Table~\ref{tab:failure-classes}) are human-curated from audited
sessions; the underlying CSV is frozen at paper-build time. To
balance verifiability against the privacy of the operator's
production wallets we report counts and aggregates rather than
hashes and addresses. The per-transaction extract behind these
tables carries the transaction IDs, but we do not release it: the
operations are the operator's own, and publishing the hashes
under a named paper would link a real wallet's full balance
history to an identity permanently. The aggregate tables are
therefore the unit a reader can check, and the on-chain record
should be read as evidence that the system runs in production
rather than as an independently verifiable result.
Adversarial-suite trials execute against freshly-generated test
wallets, never against production wallets.

\input{content/tables/primitive_breakdown}
\input{content/tables/chain_coverage}
\input{content/tables/failure_classes}
\input{content/tables/avatar_fleet}

Table~\ref{tab:primitives} breaks the on-chain record down by the
eight transition primitives of \S\ref{sec:formalism}. Every
primitive has at least one real operation behind it, with
\textsc{Bridge} and \textsc{Trade} dominating---a function of the
operator's usage pattern (cross-chain swaps and exchange activity)
rather than a design choice about which primitives to support.
Table~\ref{tab:chains} pivots the same data by source chain: eight
distinct chains / substrates, including Bitcoin (native transfers),
Hyperliquid's non-EVM L1, and Solana (a non-EVM smart-contract
chain). Several chains are L2 rollups or sidechain-style execution
layers; their specific trust surface (sequencer censorship,
withdrawal timer) is not engaged by the seven invariants, which
treat L2-only inclusion as confirmed for the user's purposes; a
sequencer-independent deployment that pushed observability down to
L1 inclusion proofs is out of scope for this paper.

Table~\ref{tab:failure-classes} enumerates every audited safety-bug
incident, tagged by the class it belongs to
(\S\ref{sec:orchestration}) and linked to the fix PR. The list is
not exhaustive---classes 1 (phantom success) and 6 (indexer drift)
are likely under-counted because early incidents self-resolved
without detailed audit---but it covers every incident where the
session ID, failure shape, and subsequent fix were recorded.

Table~\ref{tab:avatar-fleet} reports the production avatar fleet's
outcome distribution. The fleet runs continuously: a DCA template
fires on a daily schedule and a market-analysis template runs on a
news-event trigger. In the data window, $32\%$ of DCA runs end in
\textsf{success}, $32\%$ in \textsf{escalated} (the avatar deferred
to the human rather than re-issue a write), $24\%$ in
\textsf{error} (executor hard failure, correctly surfaced rather
than retried), and $11\%$ in \textsf{aborted} (safety stop). None
of the $87$ DCA runs produced a duplicated on-chain transition; the
realized double-spend in the case study below predates the
deployment of Invariant~\ref{inv:plan-retry}. Avatar runs are
non-adversarial by construction (the avatar is software the user
parameterized, not an attacker), so the share of runs the
executor or avatar tier blocks is a direct lower bound on the
false-positive rate of
Invariants~\ref{inv:postobs}--\ref{inv:plan-retry} under realistic
load.

\paragraph{Case study: Phantom failure, realized double-spend.}
On 2026-04-30 a DCA avatar issuing a 200~USDT~$\to$~BTC daily swap
saw the executor report ``failed,'' retried, and confirmed both
attempts on chain (two Ethereum transactions, three minutes apart)
for a total of 400~USDT against an intended 200. The executor-local
fix from an earlier near-miss had shipped, but the avatar tier
treated the executor's first response as authoritative at plan
granularity. Under Invariant~\ref{inv:plan-retry} the second
transaction does not happen: the avatar must positively verify the
first attempt's on-chain absence before re-issuing, which fails
here because the first attempt did land. The incident is the
empirical motivation behind that invariant; a verify-then-retry
implementation at the avatar tier shipped the same week.

\paragraph{Case study: Webhook duplicate delivery.}
On 2026-05-01 an external trigger source delivered a webhook to
the avatar service. The pod processed the request and started a
rebalance run, but the edge proxy returned 502 to the trigger
source despite upstream success. The trigger retried the same
byte-for-byte payload eight seconds later. Without idempotency the
duplicate would have either started a second run or returned 409
on the in-flight first run, orphaning the original from the
trigger source's logs. The fix shipped the same day: a
sliding-window cache keyed by $H(\textsf{sender\_id},
\textsf{request\_bytes})$ with $\Delta = 30$\,s; on the retry the
cache hits and the original response is returned verbatim. This
is the failure-class~9 incident and the motivation behind
Invariant~\ref{inv:idempotent-receiver}; the mechanism echoes
Stripe-style idempotency keys, and the contribution is the
observation that agent systems inherit at-least-once delivery
semantics from every public HTTP boundary.

\paragraph{Case study: Cross-venue prediction-market lifecycle.}
We built two parallel implementations of the prediction-market
trade tool---one per CTF-Exchange-style venue---before extracting
any shared abstraction, on the principle that one known-good
implementation is a guess at what generalizes whereas two are
data. The two-venue exercise surfaced two specifics neither
single-venue track would have caught: one venue's matching engine
enforces \texttt{takerAmount=1} as a sentinel on \emph{both sides}
of the book; the other deploys a separate Exchange contract per
market so an ERC-1155 \texttt{setApprovalForAll} preflight must be
checked per-market rather than against a global constant. Both
findings are now regression-pinned and illustrate that the
formalism's venue-agnosticism (\S\ref{sec:formalism}) does not
imply implementation-agnosticism.

%% file: content/figures/experiment_forest_claude.tex
\begin{figure}[t]
\centering
\resizebox{\linewidth}{!}{%
\begin{tikzpicture}[
  font=\small\sffamily,
  x=1mm,
  y=4.6mm,
]
  \draw[black!60] (0, 0.5) -- (100, 0.5);
  \draw[black!60] (0, 0.5) -- (0, 0.30000000000000004);
  \node[anchor=north, font=\scriptsize] at (0, 0.30000000000000004) {0\%};
  \draw[black!60] (25, 0.5) -- (25, 0.30000000000000004);
  \node[anchor=north, font=\scriptsize] at (25, 0.30000000000000004) {25\%};
  \draw[black!60] (50, 0.5) -- (50, 0.30000000000000004);
  \node[anchor=north, font=\scriptsize] at (50, 0.30000000000000004) {50\%};
  \draw[black!60] (75, 0.5) -- (75, 0.30000000000000004);
  \node[anchor=north, font=\scriptsize] at (75, 0.30000000000000004) {75\%};
  \draw[black!60] (100, 0.5) -- (100, 0.30000000000000004);
  \node[anchor=north, font=\scriptsize] at (100, 0.30000000000000004) {100\%};
  \node[anchor=north, font=\small] at (50, -0.3999999999999999) {Pass rate (95\% bootstrap CI)};
  \draw[gray!40, dashed] (81.7, 0.5) -- (81.7, 14.5);
  \node[anchor=south, font=\scriptsize, text=gray!70!black] at (81.7, 14.5) {full-stack ref};
  \fill[black!4] (0, 13.55) rectangle (100, 14.45);
  \fill[black!4] (0, 11.55) rectangle (100, 12.45);
  \fill[black!4] (0, 9.55) rectangle (100, 10.45);
  \fill[black!4] (0, 7.55) rectangle (100, 8.45);
  \fill[black!4] (0, 5.55) rectangle (100, 6.45);
  \fill[black!4] (0, 3.55) rectangle (100, 4.45);
  \fill[black!4] (0, 1.55) rectangle (100, 2.45);
  \draw[blue!60!black, line width=0.6pt] (71.7, 14) -- (91.7, 14);
  \draw[blue!60!black, line width=0.6pt] (71.7, 13.82) -- (71.7, 14.18);
  \draw[blue!60!black, line width=0.6pt] (91.7, 13.82) -- (91.7, 14.18);
  \fill[blue!60!black] (81.7, 14) circle (1.6pt);
  \node[anchor=west, font=\scriptsize] at (101.5, 14) {\texttt{claude}\hspace{0.4em}$N{=}60$\hspace{0.4em}81.7\%};
  \node[anchor=east, font=\small] at (-2, 14) {Adversarial suite};
  \draw[orange!70!black, line width=0.6pt] (100.0, 13) -- (100.0, 13);
  \draw[orange!70!black, line width=0.6pt] (100.0, 12.82) -- (100.0, 13.18);
  \draw[orange!70!black, line width=0.6pt] (100.0, 12.82) -- (100.0, 13.18);
  \fill[orange!70!black] (100.0, 13) circle (1.6pt);
  \node[anchor=west, font=\scriptsize] at (101.5, 13) {\texttt{i1\_on}\hspace{0.4em}$N{=}20$\hspace{0.4em}100.0\%};
  \node[anchor=east, font=\small] at (-2, 13) {Invariant ablation (I\textsubscript{1} confirm gate)};
  \draw[orange!70!black, line width=0.6pt] (20.0, 12) -- (60.0, 12);
  \draw[orange!70!black, line width=0.6pt] (20.0, 11.82) -- (20.0, 12.18);
  \draw[orange!70!black, line width=0.6pt] (60.0, 11.82) -- (60.0, 12.18);
  \fill[orange!70!black] (40.0, 12) circle (1.6pt);
  \node[anchor=west, font=\scriptsize] at (101.5, 12) {\texttt{i1\_off}\hspace{0.4em}$N{=}20$\hspace{0.4em}40.0\%};
  \draw[orange!70!black, line width=0.6pt] (73.3, 11) -- (91.7, 11);
  \draw[orange!70!black, line width=0.6pt] (73.3, 10.82) -- (73.3, 11.18);
  \draw[orange!70!black, line width=0.6pt] (91.7, 10.82) -- (91.7, 11.18);
  \fill[orange!70!black] (83.3, 11) circle (1.6pt);
  \node[anchor=west, font=\scriptsize] at (101.5, 11) {\texttt{claude}\hspace{0.4em}$N{=}60$\hspace{0.4em}83.3\%};
  \node[anchor=east, font=\small] at (-2, 11) {Invariant ablation (I\textsubscript{2} off)};
  \draw[orange!70!black, line width=0.6pt] (100.0, 10) -- (100.0, 10);
  \draw[orange!70!black, line width=0.6pt] (100.0, 9.82) -- (100.0, 10.18);
  \draw[orange!70!black, line width=0.6pt] (100.0, 9.82) -- (100.0, 10.18);
  \fill[orange!70!black] (100.0, 10) circle (1.6pt);
  \node[anchor=west, font=\scriptsize] at (101.5, 10) {\texttt{i3\_on}\hspace{0.4em}$N{=}20$\hspace{0.4em}100.0\%};
  \node[anchor=east, font=\small] at (-2, 10) {Invariant ablation (I\textsubscript{3} wallet-tier)};
  \draw[orange!70!black, line width=0.6pt] (10.0, 9) -- (45.0, 9);
  \draw[orange!70!black, line width=0.6pt] (10.0, 8.82) -- (10.0, 9.18);
  \draw[orange!70!black, line width=0.6pt] (45.0, 8.82) -- (45.0, 9.18);
  \fill[orange!70!black] (25.0, 9) circle (1.6pt);
  \node[anchor=west, font=\scriptsize] at (101.5, 9) {\texttt{i3\_off}\hspace{0.4em}$N{=}20$\hspace{0.4em}25.0\%};
  \draw[orange!70!black, line width=0.6pt] (73.3, 8) -- (91.7, 8);
  \draw[orange!70!black, line width=0.6pt] (73.3, 7.82) -- (73.3, 8.18);
  \draw[orange!70!black, line width=0.6pt] (91.7, 7.82) -- (91.7, 8.18);
  \fill[orange!70!black] (83.3, 8) circle (1.6pt);
  \node[anchor=west, font=\scriptsize] at (101.5, 8) {\texttt{claude}\hspace{0.4em}$N{=}60$\hspace{0.4em}83.3\%};
  \node[anchor=east, font=\small] at (-2, 8) {Invariant ablation (I\textsubscript{4} off)};
  \draw[orange!70!black, line width=0.6pt] (100.0, 7) -- (100.0, 7);
  \draw[orange!70!black, line width=0.6pt] (100.0, 6.82) -- (100.0, 7.18);
  \draw[orange!70!black, line width=0.6pt] (100.0, 6.82) -- (100.0, 7.18);
  \fill[orange!70!black] (100.0, 7) circle (1.6pt);
  \node[anchor=west, font=\scriptsize] at (101.5, 7) {\texttt{i5\_on}\hspace{0.4em}$N{=}20$\hspace{0.4em}100.0\%};
  \node[anchor=east, font=\small] at (-2, 7) {Invariant ablation (I\textsubscript{5} delegation-tier)};
  \draw[orange!70!black, line width=0.6pt] (10.0, 6) -- (45.0, 6);
  \draw[orange!70!black, line width=0.6pt] (10.0, 5.82) -- (10.0, 6.18);
  \draw[orange!70!black, line width=0.6pt] (45.0, 5.82) -- (45.0, 6.18);
  \fill[orange!70!black] (25.0, 6) circle (1.6pt);
  \node[anchor=west, font=\scriptsize] at (101.5, 6) {\texttt{i5\_off}\hspace{0.4em}$N{=}20$\hspace{0.4em}25.0\%};
  \draw[orange!70!black, line width=0.6pt] (100.0, 5) -- (100.0, 5);
  \draw[orange!70!black, line width=0.6pt] (100.0, 4.82) -- (100.0, 5.18);
  \draw[orange!70!black, line width=0.6pt] (100.0, 4.82) -- (100.0, 5.18);
  \fill[orange!70!black] (100.0, 5) circle (1.6pt);
  \node[anchor=west, font=\scriptsize] at (101.5, 5) {\texttt{i6\_on}\hspace{0.4em}$N{=}20$\hspace{0.4em}100.0\%};
  \node[anchor=east, font=\small] at (-2, 5) {Invariant ablation (I\textsubscript{6} cross-tier)};
  \draw[orange!70!black, line width=0.6pt] (0.0, 4) -- (0.0, 4);
  \draw[orange!70!black, line width=0.6pt] (0.0, 3.82) -- (0.0, 4.18);
  \draw[orange!70!black, line width=0.6pt] (0.0, 3.82) -- (0.0, 4.18);
  \fill[orange!70!black] (0.0, 4) circle (1.6pt);
  \node[anchor=west, font=\scriptsize] at (101.5, 4) {\texttt{i6\_off}\hspace{0.4em}$N{=}20$\hspace{0.4em}0.0\%};
  \draw[orange!70!black, line width=0.6pt] (100.0, 3) -- (100.0, 3);
  \draw[orange!70!black, line width=0.6pt] (100.0, 2.82) -- (100.0, 3.18);
  \draw[orange!70!black, line width=0.6pt] (100.0, 2.82) -- (100.0, 3.18);
  \fill[orange!70!black] (100.0, 3) circle (1.6pt);
  \node[anchor=west, font=\scriptsize] at (101.5, 3) {\texttt{i7\_on}\hspace{0.4em}$N{=}20$\hspace{0.4em}100.0\%};
  \node[anchor=east, font=\small] at (-2, 3) {Invariant ablation (I\textsubscript{7} cross-tier)};
  \draw[orange!70!black, line width=0.6pt] (0.0, 2) -- (0.0, 2);
  \draw[orange!70!black, line width=0.6pt] (0.0, 1.82) -- (0.0, 2.18);
  \draw[orange!70!black, line width=0.6pt] (0.0, 1.82) -- (0.0, 2.18);
  \fill[orange!70!black] (0.0, 2) circle (1.6pt);
  \node[anchor=west, font=\scriptsize] at (101.5, 2) {\texttt{i7\_off}\hspace{0.4em}$N{=}20$\hspace{0.4em}0.0\%};
  \draw[red!55!black, line width=0.6pt] (3.4, 1) -- (18.6, 1);
  \draw[red!55!black, line width=0.6pt] (3.4, 0.8200000000000001) -- (3.4, 1.18);
  \draw[red!55!black, line width=0.6pt] (18.6, 0.8200000000000001) -- (18.6, 1.18);
  \fill[red!55!black] (10.2, 1) circle (1.6pt);
  \node[anchor=west, font=\scriptsize] at (101.5, 1) {\texttt{claude}\hspace{0.4em}$N{=}59$\hspace{0.4em}10.2\%};
  \node[anchor=east, font=\small] at (-2, 1) {Baseline (naive ReAct)};
\end{tikzpicture}%
}
\caption{Experimental results with backing model Claude. Each row's whisker is a 95\% bootstrap confidence interval over 1000 resamples with replacement of the per-row Bernoulli outcomes; the filled point is the empirical pass rate. The dashed vertical reference line marks this model's full-stack adversarial pass rate so each ablation's drop reads off against a common ceiling. Code-gate ablation rows (I\textsubscript{1}, I\textsubscript{3}, I\textsubscript{5}, I\textsubscript{6}, I\textsubscript{7}) describe a gate at the executor or avatar tier rather than the backing model, and so appear identically in the companion figure for the other backing model.}
\label{fig:experiments-claude}
\end{figure}

%% file: content/figures/experiment_forest_minimax.tex
\begin{figure}[t]
\centering
\resizebox{\linewidth}{!}{%
\begin{tikzpicture}[
  font=\small\sffamily,
  x=1mm,
  y=4.6mm,
]
  \draw[black!60] (0, 0.5) -- (100, 0.5);
  \draw[black!60] (0, 0.5) -- (0, 0.30000000000000004);
  \node[anchor=north, font=\scriptsize] at (0, 0.30000000000000004) {0\%};
  \draw[black!60] (25, 0.5) -- (25, 0.30000000000000004);
  \node[anchor=north, font=\scriptsize] at (25, 0.30000000000000004) {25\%};
  \draw[black!60] (50, 0.5) -- (50, 0.30000000000000004);
  \node[anchor=north, font=\scriptsize] at (50, 0.30000000000000004) {50\%};
  \draw[black!60] (75, 0.5) -- (75, 0.30000000000000004);
  \node[anchor=north, font=\scriptsize] at (75, 0.30000000000000004) {75\%};
  \draw[black!60] (100, 0.5) -- (100, 0.30000000000000004);
  \node[anchor=north, font=\scriptsize] at (100, 0.30000000000000004) {100\%};
  \node[anchor=north, font=\small] at (50, -0.3999999999999999) {Pass rate (95\% bootstrap CI)};
  \draw[gray!40, dashed] (86.7, 0.5) -- (86.7, 14.5);
  \node[anchor=south, font=\scriptsize, text=gray!70!black] at (86.7, 14.5) {full-stack ref};
  \fill[black!4] (0, 13.55) rectangle (100, 14.45);
  \fill[black!4] (0, 11.55) rectangle (100, 12.45);
  \fill[black!4] (0, 9.55) rectangle (100, 10.45);
  \fill[black!4] (0, 7.55) rectangle (100, 8.45);
  \fill[black!4] (0, 5.55) rectangle (100, 6.45);
  \fill[black!4] (0, 3.55) rectangle (100, 4.45);
  \fill[black!4] (0, 1.55) rectangle (100, 2.45);
  \draw[blue!60!black, line width=0.6pt] (78.3, 14) -- (95.0, 14);
  \draw[blue!60!black, line width=0.6pt] (78.3, 13.82) -- (78.3, 14.18);
  \draw[blue!60!black, line width=0.6pt] (95.0, 13.82) -- (95.0, 14.18);
  \fill[blue!60!black] (86.7, 14) circle (1.6pt);
  \node[anchor=west, font=\scriptsize] at (101.5, 14) {\texttt{minimax}\hspace{0.4em}$N{=}60$\hspace{0.4em}86.7\%};
  \node[anchor=east, font=\small] at (-2, 14) {Adversarial suite};
  \draw[orange!70!black, line width=0.6pt] (100.0, 13) -- (100.0, 13);
  \draw[orange!70!black, line width=0.6pt] (100.0, 12.82) -- (100.0, 13.18);
  \draw[orange!70!black, line width=0.6pt] (100.0, 12.82) -- (100.0, 13.18);
  \fill[orange!70!black] (100.0, 13) circle (1.6pt);
  \node[anchor=west, font=\scriptsize] at (101.5, 13) {\texttt{i1\_on}\hspace{0.4em}$N{=}20$\hspace{0.4em}100.0\%};
  \node[anchor=east, font=\small] at (-2, 13) {Invariant ablation (I\textsubscript{1} confirm gate)};
  \draw[orange!70!black, line width=0.6pt] (20.0, 12) -- (60.0, 12);
  \draw[orange!70!black, line width=0.6pt] (20.0, 11.82) -- (20.0, 12.18);
  \draw[orange!70!black, line width=0.6pt] (60.0, 11.82) -- (60.0, 12.18);
  \fill[orange!70!black] (40.0, 12) circle (1.6pt);
  \node[anchor=west, font=\scriptsize] at (101.5, 12) {\texttt{i1\_off}\hspace{0.4em}$N{=}20$\hspace{0.4em}40.0\%};
  \draw[orange!70!black, line width=0.6pt] (78.3, 11) -- (95.0, 11);
  \draw[orange!70!black, line width=0.6pt] (78.3, 10.82) -- (78.3, 11.18);
  \draw[orange!70!black, line width=0.6pt] (95.0, 10.82) -- (95.0, 11.18);
  \fill[orange!70!black] (86.7, 11) circle (1.6pt);
  \node[anchor=west, font=\scriptsize] at (101.5, 11) {\texttt{minimax}\hspace{0.4em}$N{=}60$\hspace{0.4em}86.7\%};
  \node[anchor=east, font=\small] at (-2, 11) {Invariant ablation (I\textsubscript{2} off)};
  \draw[orange!70!black, line width=0.6pt] (100.0, 10) -- (100.0, 10);
  \draw[orange!70!black, line width=0.6pt] (100.0, 9.82) -- (100.0, 10.18);
  \draw[orange!70!black, line width=0.6pt] (100.0, 9.82) -- (100.0, 10.18);
  \fill[orange!70!black] (100.0, 10) circle (1.6pt);
  \node[anchor=west, font=\scriptsize] at (101.5, 10) {\texttt{i3\_on}\hspace{0.4em}$N{=}20$\hspace{0.4em}100.0\%};
  \node[anchor=east, font=\small] at (-2, 10) {Invariant ablation (I\textsubscript{3} wallet-tier)};
  \draw[orange!70!black, line width=0.6pt] (10.0, 9) -- (45.0, 9);
  \draw[orange!70!black, line width=0.6pt] (10.0, 8.82) -- (10.0, 9.18);
  \draw[orange!70!black, line width=0.6pt] (45.0, 8.82) -- (45.0, 9.18);
  \fill[orange!70!black] (25.0, 9) circle (1.6pt);
  \node[anchor=west, font=\scriptsize] at (101.5, 9) {\texttt{i3\_off}\hspace{0.4em}$N{=}20$\hspace{0.4em}25.0\%};
  \draw[orange!70!black, line width=0.6pt] (75.0, 8) -- (93.3, 8);
  \draw[orange!70!black, line width=0.6pt] (75.0, 7.82) -- (75.0, 8.18);
  \draw[orange!70!black, line width=0.6pt] (93.3, 7.82) -- (93.3, 8.18);
  \fill[orange!70!black] (85.0, 8) circle (1.6pt);
  \node[anchor=west, font=\scriptsize] at (101.5, 8) {\texttt{minimax}\hspace{0.4em}$N{=}60$\hspace{0.4em}85.0\%};
  \node[anchor=east, font=\small] at (-2, 8) {Invariant ablation (I\textsubscript{4} off)};
  \draw[orange!70!black, line width=0.6pt] (100.0, 7) -- (100.0, 7);
  \draw[orange!70!black, line width=0.6pt] (100.0, 6.82) -- (100.0, 7.18);
  \draw[orange!70!black, line width=0.6pt] (100.0, 6.82) -- (100.0, 7.18);
  \fill[orange!70!black] (100.0, 7) circle (1.6pt);
  \node[anchor=west, font=\scriptsize] at (101.5, 7) {\texttt{i5\_on}\hspace{0.4em}$N{=}20$\hspace{0.4em}100.0\%};
  \node[anchor=east, font=\small] at (-2, 7) {Invariant ablation (I\textsubscript{5} delegation-tier)};
  \draw[orange!70!black, line width=0.6pt] (10.0, 6) -- (45.0, 6);
  \draw[orange!70!black, line width=0.6pt] (10.0, 5.82) -- (10.0, 6.18);
  \draw[orange!70!black, line width=0.6pt] (45.0, 5.82) -- (45.0, 6.18);
  \fill[orange!70!black] (25.0, 6) circle (1.6pt);
  \node[anchor=west, font=\scriptsize] at (101.5, 6) {\texttt{i5\_off}\hspace{0.4em}$N{=}20$\hspace{0.4em}25.0\%};
  \draw[orange!70!black, line width=0.6pt] (100.0, 5) -- (100.0, 5);
  \draw[orange!70!black, line width=0.6pt] (100.0, 4.82) -- (100.0, 5.18);
  \draw[orange!70!black, line width=0.6pt] (100.0, 4.82) -- (100.0, 5.18);
  \fill[orange!70!black] (100.0, 5) circle (1.6pt);
  \node[anchor=west, font=\scriptsize] at (101.5, 5) {\texttt{i6\_on}\hspace{0.4em}$N{=}20$\hspace{0.4em}100.0\%};
  \node[anchor=east, font=\small] at (-2, 5) {Invariant ablation (I\textsubscript{6} cross-tier)};
  \draw[orange!70!black, line width=0.6pt] (0.0, 4) -- (0.0, 4);
  \draw[orange!70!black, line width=0.6pt] (0.0, 3.82) -- (0.0, 4.18);
  \draw[orange!70!black, line width=0.6pt] (0.0, 3.82) -- (0.0, 4.18);
  \fill[orange!70!black] (0.0, 4) circle (1.6pt);
  \node[anchor=west, font=\scriptsize] at (101.5, 4) {\texttt{i6\_off}\hspace{0.4em}$N{=}20$\hspace{0.4em}0.0\%};
  \draw[orange!70!black, line width=0.6pt] (100.0, 3) -- (100.0, 3);
  \draw[orange!70!black, line width=0.6pt] (100.0, 2.82) -- (100.0, 3.18);
  \draw[orange!70!black, line width=0.6pt] (100.0, 2.82) -- (100.0, 3.18);
  \fill[orange!70!black] (100.0, 3) circle (1.6pt);
  \node[anchor=west, font=\scriptsize] at (101.5, 3) {\texttt{i7\_on}\hspace{0.4em}$N{=}20$\hspace{0.4em}100.0\%};
  \node[anchor=east, font=\small] at (-2, 3) {Invariant ablation (I\textsubscript{7} cross-tier)};
  \draw[orange!70!black, line width=0.6pt] (0.0, 2) -- (0.0, 2);
  \draw[orange!70!black, line width=0.6pt] (0.0, 1.82) -- (0.0, 2.18);
  \draw[orange!70!black, line width=0.6pt] (0.0, 1.82) -- (0.0, 2.18);
  \fill[orange!70!black] (0.0, 2) circle (1.6pt);
  \node[anchor=west, font=\scriptsize] at (101.5, 2) {\texttt{i7\_off}\hspace{0.4em}$N{=}20$\hspace{0.4em}0.0\%};
  \draw[red!55!black, line width=0.6pt] (3.3, 1) -- (18.3, 1);
  \draw[red!55!black, line width=0.6pt] (3.3, 0.8200000000000001) -- (3.3, 1.18);
  \draw[red!55!black, line width=0.6pt] (18.3, 0.8200000000000001) -- (18.3, 1.18);
  \fill[red!55!black] (10.0, 1) circle (1.6pt);
  \node[anchor=west, font=\scriptsize] at (101.5, 1) {\texttt{minimax}\hspace{0.4em}$N{=}60$\hspace{0.4em}10.0\%};
  \node[anchor=east, font=\small] at (-2, 1) {Baseline (naive ReAct)};
\end{tikzpicture}%
}
\caption{Experimental results with backing model MiniMax-M2.7-highspeed. Each row's whisker is a 95\% bootstrap confidence interval over 1000 resamples with replacement of the per-row Bernoulli outcomes; the filled point is the empirical pass rate. The dashed vertical reference line marks this model's full-stack adversarial pass rate so each ablation's drop reads off against a common ceiling. Code-gate ablation rows (I\textsubscript{1}, I\textsubscript{3}, I\textsubscript{5}, I\textsubscript{6}, I\textsubscript{7}) describe a gate at the executor or avatar tier rather than the backing model, and so appear identically in the companion figure for the other backing model.}
\label{fig:experiments-minimax}
\end{figure}

%% file: content/figures/experiment_forest_gemini.tex
\begin{figure}[t]
\centering
\resizebox{\linewidth}{!}{%
\begin{tikzpicture}[
  font=\small\sffamily,
  x=1mm,
  y=4.6mm,
]
  \draw[black!60] (0, 0.5) -- (100, 0.5);
  \draw[black!60] (0, 0.5) -- (0, 0.30000000000000004);
  \node[anchor=north, font=\scriptsize] at (0, 0.30000000000000004) {0\%};
  \draw[black!60] (25, 0.5) -- (25, 0.30000000000000004);
  \node[anchor=north, font=\scriptsize] at (25, 0.30000000000000004) {25\%};
  \draw[black!60] (50, 0.5) -- (50, 0.30000000000000004);
  \node[anchor=north, font=\scriptsize] at (50, 0.30000000000000004) {50\%};
  \draw[black!60] (75, 0.5) -- (75, 0.30000000000000004);
  \node[anchor=north, font=\scriptsize] at (75, 0.30000000000000004) {75\%};
  \draw[black!60] (100, 0.5) -- (100, 0.30000000000000004);
  \node[anchor=north, font=\scriptsize] at (100, 0.30000000000000004) {100\%};
  \node[anchor=north, font=\small] at (50, -0.3999999999999999) {Pass rate (95\% bootstrap CI)};
  \draw[gray!40, dashed] (60.0, 0.5) -- (60.0, 14.5);
  \node[anchor=south, font=\scriptsize, text=gray!70!black] at (60.0, 14.5) {full-stack ref};
  \fill[black!4] (0, 13.55) rectangle (100, 14.45);
  \fill[black!4] (0, 11.55) rectangle (100, 12.45);
  \fill[black!4] (0, 9.55) rectangle (100, 10.45);
  \fill[black!4] (0, 7.55) rectangle (100, 8.45);
  \fill[black!4] (0, 5.55) rectangle (100, 6.45);
  \fill[black!4] (0, 3.55) rectangle (100, 4.45);
  \fill[black!4] (0, 1.55) rectangle (100, 2.45);
  \draw[blue!60!black, line width=0.6pt] (48.3, 14) -- (71.7, 14);
  \draw[blue!60!black, line width=0.6pt] (48.3, 13.82) -- (48.3, 14.18);
  \draw[blue!60!black, line width=0.6pt] (71.7, 13.82) -- (71.7, 14.18);
  \fill[blue!60!black] (60.0, 14) circle (1.6pt);
  \node[anchor=west, font=\scriptsize] at (101.5, 14) {\texttt{gemini}\hspace{0.4em}$N{=}60$\hspace{0.4em}60.0\%};
  \node[anchor=east, font=\small] at (-2, 14) {Adversarial suite};
  \draw[orange!70!black, line width=0.6pt] (100.0, 13) -- (100.0, 13);
  \draw[orange!70!black, line width=0.6pt] (100.0, 12.82) -- (100.0, 13.18);
  \draw[orange!70!black, line width=0.6pt] (100.0, 12.82) -- (100.0, 13.18);
  \fill[orange!70!black] (100.0, 13) circle (1.6pt);
  \node[anchor=west, font=\scriptsize] at (101.5, 13) {\texttt{i1\_on}\hspace{0.4em}$N{=}20$\hspace{0.4em}100.0\%};
  \node[anchor=east, font=\small] at (-2, 13) {Invariant ablation (I\textsubscript{1} confirm gate)};
  \draw[orange!70!black, line width=0.6pt] (20.0, 12) -- (60.0, 12);
  \draw[orange!70!black, line width=0.6pt] (20.0, 11.82) -- (20.0, 12.18);
  \draw[orange!70!black, line width=0.6pt] (60.0, 11.82) -- (60.0, 12.18);
  \fill[orange!70!black] (40.0, 12) circle (1.6pt);
  \node[anchor=west, font=\scriptsize] at (101.5, 12) {\texttt{i1\_off}\hspace{0.4em}$N{=}20$\hspace{0.4em}40.0\%};
  \draw[orange!70!black, line width=0.6pt] (50.0, 11) -- (73.3, 11);
  \draw[orange!70!black, line width=0.6pt] (50.0, 10.82) -- (50.0, 11.18);
  \draw[orange!70!black, line width=0.6pt] (73.3, 10.82) -- (73.3, 11.18);
  \fill[orange!70!black] (61.7, 11) circle (1.6pt);
  \node[anchor=west, font=\scriptsize] at (101.5, 11) {\texttt{gemini}\hspace{0.4em}$N{=}60$\hspace{0.4em}61.7\%};
  \node[anchor=east, font=\small] at (-2, 11) {Invariant ablation (I\textsubscript{2} off)};
  \draw[orange!70!black, line width=0.6pt] (100.0, 10) -- (100.0, 10);
  \draw[orange!70!black, line width=0.6pt] (100.0, 9.82) -- (100.0, 10.18);
  \draw[orange!70!black, line width=0.6pt] (100.0, 9.82) -- (100.0, 10.18);
  \fill[orange!70!black] (100.0, 10) circle (1.6pt);
  \node[anchor=west, font=\scriptsize] at (101.5, 10) {\texttt{i3\_on}\hspace{0.4em}$N{=}20$\hspace{0.4em}100.0\%};
  \node[anchor=east, font=\small] at (-2, 10) {Invariant ablation (I\textsubscript{3} wallet-tier)};
  \draw[orange!70!black, line width=0.6pt] (10.0, 9) -- (45.0, 9);
  \draw[orange!70!black, line width=0.6pt] (10.0, 8.82) -- (10.0, 9.18);
  \draw[orange!70!black, line width=0.6pt] (45.0, 8.82) -- (45.0, 9.18);
  \fill[orange!70!black] (25.0, 9) circle (1.6pt);
  \node[anchor=west, font=\scriptsize] at (101.5, 9) {\texttt{i3\_off}\hspace{0.4em}$N{=}20$\hspace{0.4em}25.0\%};
  \draw[orange!70!black, line width=0.6pt] (45.0, 8) -- (68.3, 8);
  \draw[orange!70!black, line width=0.6pt] (45.0, 7.82) -- (45.0, 8.18);
  \draw[orange!70!black, line width=0.6pt] (68.3, 7.82) -- (68.3, 8.18);
  \fill[orange!70!black] (56.7, 8) circle (1.6pt);
  \node[anchor=west, font=\scriptsize] at (101.5, 8) {\texttt{gemini}\hspace{0.4em}$N{=}60$\hspace{0.4em}56.7\%};
  \node[anchor=east, font=\small] at (-2, 8) {Invariant ablation (I\textsubscript{4} off)};
  \draw[orange!70!black, line width=0.6pt] (100.0, 7) -- (100.0, 7);
  \draw[orange!70!black, line width=0.6pt] (100.0, 6.82) -- (100.0, 7.18);
  \draw[orange!70!black, line width=0.6pt] (100.0, 6.82) -- (100.0, 7.18);
  \fill[orange!70!black] (100.0, 7) circle (1.6pt);
  \node[anchor=west, font=\scriptsize] at (101.5, 7) {\texttt{i5\_on}\hspace{0.4em}$N{=}20$\hspace{0.4em}100.0\%};
  \node[anchor=east, font=\small] at (-2, 7) {Invariant ablation (I\textsubscript{5} delegation-tier)};
  \draw[orange!70!black, line width=0.6pt] (10.0, 6) -- (45.0, 6);
  \draw[orange!70!black, line width=0.6pt] (10.0, 5.82) -- (10.0, 6.18);
  \draw[orange!70!black, line width=0.6pt] (45.0, 5.82) -- (45.0, 6.18);
  \fill[orange!70!black] (25.0, 6) circle (1.6pt);
  \node[anchor=west, font=\scriptsize] at (101.5, 6) {\texttt{i5\_off}\hspace{0.4em}$N{=}20$\hspace{0.4em}25.0\%};
  \draw[orange!70!black, line width=0.6pt] (100.0, 5) -- (100.0, 5);
  \draw[orange!70!black, line width=0.6pt] (100.0, 4.82) -- (100.0, 5.18);
  \draw[orange!70!black, line width=0.6pt] (100.0, 4.82) -- (100.0, 5.18);
  \fill[orange!70!black] (100.0, 5) circle (1.6pt);
  \node[anchor=west, font=\scriptsize] at (101.5, 5) {\texttt{i6\_on}\hspace{0.4em}$N{=}20$\hspace{0.4em}100.0\%};
  \node[anchor=east, font=\small] at (-2, 5) {Invariant ablation (I\textsubscript{6} cross-tier)};
  \draw[orange!70!black, line width=0.6pt] (0.0, 4) -- (0.0, 4);
  \draw[orange!70!black, line width=0.6pt] (0.0, 3.82) -- (0.0, 4.18);
  \draw[orange!70!black, line width=0.6pt] (0.0, 3.82) -- (0.0, 4.18);
  \fill[orange!70!black] (0.0, 4) circle (1.6pt);
  \node[anchor=west, font=\scriptsize] at (101.5, 4) {\texttt{i6\_off}\hspace{0.4em}$N{=}20$\hspace{0.4em}0.0\%};
  \draw[orange!70!black, line width=0.6pt] (100.0, 3) -- (100.0, 3);
  \draw[orange!70!black, line width=0.6pt] (100.0, 2.82) -- (100.0, 3.18);
  \draw[orange!70!black, line width=0.6pt] (100.0, 2.82) -- (100.0, 3.18);
  \fill[orange!70!black] (100.0, 3) circle (1.6pt);
  \node[anchor=west, font=\scriptsize] at (101.5, 3) {\texttt{i7\_on}\hspace{0.4em}$N{=}20$\hspace{0.4em}100.0\%};
  \node[anchor=east, font=\small] at (-2, 3) {Invariant ablation (I\textsubscript{7} cross-tier)};
  \draw[orange!70!black, line width=0.6pt] (0.0, 2) -- (0.0, 2);
  \draw[orange!70!black, line width=0.6pt] (0.0, 1.82) -- (0.0, 2.18);
  \draw[orange!70!black, line width=0.6pt] (0.0, 1.82) -- (0.0, 2.18);
  \fill[orange!70!black] (0.0, 2) circle (1.6pt);
  \node[anchor=west, font=\scriptsize] at (101.5, 2) {\texttt{i7\_off}\hspace{0.4em}$N{=}20$\hspace{0.4em}0.0\%};
  \draw[red!55!black, line width=0.6pt] (45.0, 1) -- (71.7, 1);
  \draw[red!55!black, line width=0.6pt] (45.0, 0.8200000000000001) -- (45.0, 1.18);
  \draw[red!55!black, line width=0.6pt] (71.7, 0.8200000000000001) -- (71.7, 1.18);
  \fill[red!55!black] (56.7, 1) circle (1.6pt);
  \node[anchor=west, font=\scriptsize] at (101.5, 1) {\texttt{gemini}\hspace{0.4em}$N{=}60$\hspace{0.4em}56.7\%};
  \node[anchor=east, font=\small] at (-2, 1) {Baseline (naive ReAct)};
\end{tikzpicture}%
}
\caption{Experimental results with backing model Gemini-2.5-flash-lite. Each row's whisker is a 95\% bootstrap confidence interval over 1000 resamples with replacement of the per-row Bernoulli outcomes; the filled point is the empirical pass rate. The dashed vertical reference line marks this model's full-stack adversarial pass rate so each ablation's drop reads off against a common ceiling. Code-gate ablation rows (I\textsubscript{1}, I\textsubscript{3}, I\textsubscript{5}, I\textsubscript{6}, I\textsubscript{7}) describe a gate at the executor or avatar tier rather than the backing model, and so appear identically in the companion figure for the other backing model.}
\label{fig:experiments-gemini}
\end{figure}

%% file: content/figures/experiment_forest_openai_nano.tex
\begin{figure}[t]
\centering
\resizebox{\linewidth}{!}{%
\begin{tikzpicture}[
  font=\small\sffamily,
  x=1mm,
  y=4.6mm,
]
  \draw[black!60] (0, 0.5) -- (100, 0.5);
  \draw[black!60] (0, 0.5) -- (0, 0.30000000000000004);
  \node[anchor=north, font=\scriptsize] at (0, 0.30000000000000004) {0\%};
  \draw[black!60] (25, 0.5) -- (25, 0.30000000000000004);
  \node[anchor=north, font=\scriptsize] at (25, 0.30000000000000004) {25\%};
  \draw[black!60] (50, 0.5) -- (50, 0.30000000000000004);
  \node[anchor=north, font=\scriptsize] at (50, 0.30000000000000004) {50\%};
  \draw[black!60] (75, 0.5) -- (75, 0.30000000000000004);
  \node[anchor=north, font=\scriptsize] at (75, 0.30000000000000004) {75\%};
  \draw[black!60] (100, 0.5) -- (100, 0.30000000000000004);
  \node[anchor=north, font=\scriptsize] at (100, 0.30000000000000004) {100\%};
  \node[anchor=north, font=\small] at (50, -0.3999999999999999) {Pass rate (95\% bootstrap CI)};
  \draw[gray!40, dashed] (36.8, 0.5) -- (36.8, 14.5);
  \node[anchor=south, font=\scriptsize, text=gray!70!black] at (36.8, 14.5) {full-stack ref};
  \fill[black!4] (0, 13.55) rectangle (100, 14.45);
  \fill[black!4] (0, 11.55) rectangle (100, 12.45);
  \fill[black!4] (0, 9.55) rectangle (100, 10.45);
  \fill[black!4] (0, 7.55) rectangle (100, 8.45);
  \fill[black!4] (0, 5.55) rectangle (100, 6.45);
  \fill[black!4] (0, 3.55) rectangle (100, 4.45);
  \fill[black!4] (0, 1.55) rectangle (100, 2.45);
  \draw[blue!60!black, line width=0.6pt] (15.8, 14) -- (57.9, 14);
  \draw[blue!60!black, line width=0.6pt] (15.8, 13.82) -- (15.8, 14.18);
  \draw[blue!60!black, line width=0.6pt] (57.9, 13.82) -- (57.9, 14.18);
  \fill[blue!60!black] (36.8, 14) circle (1.6pt);
  \node[anchor=west, font=\scriptsize] at (101.5, 14) {\texttt{openai-nano}\hspace{0.4em}$N{=}19$\hspace{0.4em}36.8\%};
  \node[anchor=east, font=\small] at (-2, 14) {Adversarial suite};
  \draw[orange!70!black, line width=0.6pt] (100.0, 13) -- (100.0, 13);
  \draw[orange!70!black, line width=0.6pt] (100.0, 12.82) -- (100.0, 13.18);
  \draw[orange!70!black, line width=0.6pt] (100.0, 12.82) -- (100.0, 13.18);
  \fill[orange!70!black] (100.0, 13) circle (1.6pt);
  \node[anchor=west, font=\scriptsize] at (101.5, 13) {\texttt{i1\_on}\hspace{0.4em}$N{=}20$\hspace{0.4em}100.0\%};
  \node[anchor=east, font=\small] at (-2, 13) {Invariant ablation (I\textsubscript{1} confirm gate)};
  \draw[orange!70!black, line width=0.6pt] (20.0, 12) -- (60.0, 12);
  \draw[orange!70!black, line width=0.6pt] (20.0, 11.82) -- (20.0, 12.18);
  \draw[orange!70!black, line width=0.6pt] (60.0, 11.82) -- (60.0, 12.18);
  \fill[orange!70!black] (40.0, 12) circle (1.6pt);
  \node[anchor=west, font=\scriptsize] at (101.5, 12) {\texttt{i1\_off}\hspace{0.4em}$N{=}20$\hspace{0.4em}40.0\%};
  \draw[orange!70!black, line width=0.6pt] (13.6, 11) -- (50.0, 11);
  \draw[orange!70!black, line width=0.6pt] (13.6, 10.82) -- (13.6, 11.18);
  \draw[orange!70!black, line width=0.6pt] (50.0, 10.82) -- (50.0, 11.18);
  \fill[orange!70!black] (31.8, 11) circle (1.6pt);
  \node[anchor=west, font=\scriptsize] at (101.5, 11) {\texttt{openai-nano}\hspace{0.4em}$N{=}22$\hspace{0.4em}31.8\%};
  \node[anchor=east, font=\small] at (-2, 11) {Invariant ablation (I\textsubscript{2} off)};
  \draw[orange!70!black, line width=0.6pt] (100.0, 10) -- (100.0, 10);
  \draw[orange!70!black, line width=0.6pt] (100.0, 9.82) -- (100.0, 10.18);
  \draw[orange!70!black, line width=0.6pt] (100.0, 9.82) -- (100.0, 10.18);
  \fill[orange!70!black] (100.0, 10) circle (1.6pt);
  \node[anchor=west, font=\scriptsize] at (101.5, 10) {\texttt{i3\_on}\hspace{0.4em}$N{=}20$\hspace{0.4em}100.0\%};
  \node[anchor=east, font=\small] at (-2, 10) {Invariant ablation (I\textsubscript{3} wallet-tier)};
  \draw[orange!70!black, line width=0.6pt] (10.0, 9) -- (45.0, 9);
  \draw[orange!70!black, line width=0.6pt] (10.0, 8.82) -- (10.0, 9.18);
  \draw[orange!70!black, line width=0.6pt] (45.0, 8.82) -- (45.0, 9.18);
  \fill[orange!70!black] (25.0, 9) circle (1.6pt);
  \node[anchor=west, font=\scriptsize] at (101.5, 9) {\texttt{i3\_off}\hspace{0.4em}$N{=}20$\hspace{0.4em}25.0\%};
  \draw[orange!70!black, line width=0.6pt] (20.0, 8) -- (60.0, 8);
  \draw[orange!70!black, line width=0.6pt] (20.0, 7.82) -- (20.0, 8.18);
  \draw[orange!70!black, line width=0.6pt] (60.0, 7.82) -- (60.0, 8.18);
  \fill[orange!70!black] (40.0, 8) circle (1.6pt);
  \node[anchor=west, font=\scriptsize] at (101.5, 8) {\texttt{openai-nano}\hspace{0.4em}$N{=}20$\hspace{0.4em}40.0\%};
  \node[anchor=east, font=\small] at (-2, 8) {Invariant ablation (I\textsubscript{4} off)};
  \draw[orange!70!black, line width=0.6pt] (100.0, 7) -- (100.0, 7);
  \draw[orange!70!black, line width=0.6pt] (100.0, 6.82) -- (100.0, 7.18);
  \draw[orange!70!black, line width=0.6pt] (100.0, 6.82) -- (100.0, 7.18);
  \fill[orange!70!black] (100.0, 7) circle (1.6pt);
  \node[anchor=west, font=\scriptsize] at (101.5, 7) {\texttt{i5\_on}\hspace{0.4em}$N{=}20$\hspace{0.4em}100.0\%};
  \node[anchor=east, font=\small] at (-2, 7) {Invariant ablation (I\textsubscript{5} delegation-tier)};
  \draw[orange!70!black, line width=0.6pt] (10.0, 6) -- (45.0, 6);
  \draw[orange!70!black, line width=0.6pt] (10.0, 5.82) -- (10.0, 6.18);
  \draw[orange!70!black, line width=0.6pt] (45.0, 5.82) -- (45.0, 6.18);
  \fill[orange!70!black] (25.0, 6) circle (1.6pt);
  \node[anchor=west, font=\scriptsize] at (101.5, 6) {\texttt{i5\_off}\hspace{0.4em}$N{=}20$\hspace{0.4em}25.0\%};
  \draw[orange!70!black, line width=0.6pt] (100.0, 5) -- (100.0, 5);
  \draw[orange!70!black, line width=0.6pt] (100.0, 4.82) -- (100.0, 5.18);
  \draw[orange!70!black, line width=0.6pt] (100.0, 4.82) -- (100.0, 5.18);
  \fill[orange!70!black] (100.0, 5) circle (1.6pt);
  \node[anchor=west, font=\scriptsize] at (101.5, 5) {\texttt{i6\_on}\hspace{0.4em}$N{=}20$\hspace{0.4em}100.0\%};
  \node[anchor=east, font=\small] at (-2, 5) {Invariant ablation (I\textsubscript{6} cross-tier)};
  \draw[orange!70!black, line width=0.6pt] (0.0, 4) -- (0.0, 4);
  \draw[orange!70!black, line width=0.6pt] (0.0, 3.82) -- (0.0, 4.18);
  \draw[orange!70!black, line width=0.6pt] (0.0, 3.82) -- (0.0, 4.18);
  \fill[orange!70!black] (0.0, 4) circle (1.6pt);
  \node[anchor=west, font=\scriptsize] at (101.5, 4) {\texttt{i6\_off}\hspace{0.4em}$N{=}20$\hspace{0.4em}0.0\%};
  \draw[orange!70!black, line width=0.6pt] (100.0, 3) -- (100.0, 3);
  \draw[orange!70!black, line width=0.6pt] (100.0, 2.82) -- (100.0, 3.18);
  \draw[orange!70!black, line width=0.6pt] (100.0, 2.82) -- (100.0, 3.18);
  \fill[orange!70!black] (100.0, 3) circle (1.6pt);
  \node[anchor=west, font=\scriptsize] at (101.5, 3) {\texttt{i7\_on}\hspace{0.4em}$N{=}20$\hspace{0.4em}100.0\%};
  \node[anchor=east, font=\small] at (-2, 3) {Invariant ablation (I\textsubscript{7} cross-tier)};
  \draw[orange!70!black, line width=0.6pt] (0.0, 2) -- (0.0, 2);
  \draw[orange!70!black, line width=0.6pt] (0.0, 1.82) -- (0.0, 2.18);
  \draw[orange!70!black, line width=0.6pt] (0.0, 1.82) -- (0.0, 2.18);
  \fill[orange!70!black] (0.0, 2) circle (1.6pt);
  \node[anchor=west, font=\scriptsize] at (101.5, 2) {\texttt{i7\_off}\hspace{0.4em}$N{=}20$\hspace{0.4em}0.0\%};
  \draw[red!55!black, line width=0.6pt] (20.0, 1) -- (43.3, 1);
  \draw[red!55!black, line width=0.6pt] (20.0, 0.8200000000000001) -- (20.0, 1.18);
  \draw[red!55!black, line width=0.6pt] (43.3, 0.8200000000000001) -- (43.3, 1.18);
  \fill[red!55!black] (31.7, 1) circle (1.6pt);
  \node[anchor=west, font=\scriptsize] at (101.5, 1) {\texttt{openai-nano}\hspace{0.4em}$N{=}60$\hspace{0.4em}31.7\%};
  \node[anchor=east, font=\small] at (-2, 1) {Baseline (naive ReAct)};
\end{tikzpicture}%
}
\caption{Experimental results with backing model GPT-5-nano. Each row's whisker is a 95\% bootstrap confidence interval over 1000 resamples with replacement of the per-row Bernoulli outcomes; the filled point is the empirical pass rate. The dashed vertical reference line marks this model's full-stack adversarial pass rate so each ablation's drop reads off against a common ceiling. Code-gate ablation rows (I\textsubscript{1}, I\textsubscript{3}, I\textsubscript{5}, I\textsubscript{6}, I\textsubscript{7}) describe a gate at the executor or avatar tier rather than the backing model, and so appear identically in the companion figure for the other backing model.}
\label{fig:experiments-openai-nano}
\end{figure}

%% file: content/tables/primitive_breakdown.tex
\begin{table}[t]
\centering
\caption{On-chain write operations by transition primitive (\S3). Each row is a transaction that landed on a public ledger or was accepted by the Hyperliquid exchange API. 108 total operations.}
\label{tab:primitives}
\begin{tabular}{lr}
\toprule
Primitive & Count \\
\midrule
Transfer & 12 \\
Swap & 16 \\
Bridge & 27 \\
Deposit & 1 \\
Withdraw & 3 \\
Trade & 35 \\
Intra-Exchange & 5 \\
Contract & 9 \\
\midrule
\textbf{Total} & \textbf{108} \\
\bottomrule
\end{tabular}
\end{table}

%% file: content/tables/chain_coverage.tex
\begin{table}[t]
\centering
\caption{On-chain coverage: operations per source chain. Bridge operations are counted by their source chain, the side that the agent signed on.}
\label{tab:chains}
\begin{tabular}{lr}
\toprule
Chain & Operations \\
\midrule
ethereum & 35 \\
base & 29 \\
polygon & 19 \\
hyperliquid & 16 \\
solana & 4 \\
bitcoin & 2 \\
bsc & 2 \\
arbitrum & 1 \\
\midrule
\textbf{Total} & \textbf{108} \\
\bottomrule
\end{tabular}
\end{table}

%% file: content/tables/failure_classes.tex
\begin{table}[t]
\centering
\caption{Safety bug incidents observed in the wild, grouped by the eight failure classes of \S\ref{sec:orchestration}. Every class has at least one audited incident with a corresponding fix shipped to production; the case-study paragraphs in this section cite the specific incidents where the detail matters.}
\label{tab:failure-classes}
\footnotesize
\begin{tabular}{@{}cp{0.55\columnwidth}c@{}}
\toprule
Class & Description & Fix shipped \\
\midrule
1 & Phantom success & \cmark \\
2 & Phantom failure & \cmark \\
3 & Preview gate bypass & \cmark \\
4 & Wrong-wallet selection & \cmark \\
5 & Provider gap (THORChain no-USDC-on-BSC) & \cmark \\
6 & Slow indexer drift & \cmark \\
7 & Stale agent context & \cmark \\
8 & Documented-atomicity drift & \cmark \\
9 & Duplicate webhook delivery (TradingView 502 retry) & \cmark \\
\bottomrule
\end{tabular}
\end{table}

%% file: content/tables/avatar_fleet.tex
\begin{table}[t]
\centering
\caption{Production avatar-fleet outcome distribution. One row per installed template; columns are the four terminal-state outcomes recorded in the avatar service's \texttt{avatar\_runs} table over the fleet's lifetime to date. Avatar runs are non-adversarial by construction (the avatar is software the user installed and parameterized, not an attacker), so the fraction of runs that the executor blocks is a direct lower bound on Invariants \ref{inv:postobs}--\ref{inv:plan-retry}'s false-positive rate under realistic load.}
\label{tab:avatar-fleet}
\small
\begin{tabular}{lrrrrr}
\toprule
Template & Success & Escalated & Error & Aborted & Total \\
\midrule
\textsf{dca} & 28 & 28 & 21 & 10 & 87 \\
\textsf{market\_analysis} & 2 & 0 & 0 & 0 & 2 \\
\midrule
\textbf{All} & 30 & 28 & 21 & 10 & \textbf{89} \\
\bottomrule
\end{tabular}
\end{table}

%% file: content/06_related_work.tex
\section{Related Work}
\label{sec:related}

\paragraph{Tool-use in agent frameworks and long-context drift.}
ReAct~\cite{yao2023react} interleaves reasoning and tool
invocation; Reflexion~\cite{shinn2023reflexion} adds verbal
self-correction; the Model Context Protocol (MCP)~\cite{mcp2024}
standardizes tool wire formats. These frameworks optimize
\emph{task success} on benchmarks and give little attention to
irreversible side-effects; a recent
systematization~\cite{mcp_sok} catalogs MCP-layer risks (tool
poisoning, prompt injection, epistemic failures) but does not
address the coercion-at-the-type-boundary failure we observe in
class~3. Long-context behavior is documented as ``lost in the
middle''~\cite{liu2023lost}; subsequent work
sharpens the phenomenon at the
start~\cite{lost_beginning_reasoning} and under intent
mismatch~\cite{intent_mismatch_lost,laban2025llms}. Our
\emph{stale-context} failure (class 7) is a specific instance:
the stale conclusion is whether a write is possible at all,
addressed structurally (smart wallet-active picker + skill-level
fresh-query rules on repeat write intents) rather than by
prompting alone. The hallucination
literature~\cite{ji2023survey,huang2023survey} is adjacent but
distinct: our failure is not fabrication, it is over-weighting a
previously-correct observation.

\paragraph{Atomicity and end-to-end precondition checking.}
Classical transaction semantics~\cite{gray1993transaction,
herlihy1990linearizability} give precise definitions financial
systems are expected to uphold. In practice venue APIs routinely
expose documented semantics that diverge from implementation
(class 8 in our taxonomy), broadly consistent with the atomicity
pitfalls documented in DeFi~\cite{daian2020flash,zhou2023sok}.
Defending against documented-vs-implemented drift reduces to
precondition checking on the client side rather than trust in
the documented guarantee---an application of the end-to-end
argument~\cite{saltzer1984end} to order-replacement APIs.

\paragraph{Multi-agent systems, capabilities, and account
abstraction.}
Multi-agent frameworks (AutoGen~\cite{wu2023autogen},
MetaGPT~\cite{hong2023metagpt}) formalize multi-LLM
collaboration via structured conversation; our orchestration
(\S\ref{sec:multi-tier}) layers over a single-actor surface
unchanged by the multi-agent setting---a human user and an
avatar are interchangeable from the executor's point of view,
which the multi-agent literature does not address. Invariant
\ref{inv:delegation} borrows vocabulary from object-capability
security~\cite{levy1984capability,miller2006robust,
miller2003capabilities}; we instantiate the surface (scoped
bearer tokens, audit ground truth) but do not prove the formal
object-capability properties, flagged as future work
(\S\ref{sec:discussion}). The natural alternative substrate is
ERC-4337 account abstraction~\cite{erc4337}, where smart-contract
wallets expose session keys (Safe~\cite{safe_aa_modules},
Argent~\cite{argent_session_keys}) that programmatically restrict
a delegated signer with on-chain revocation transparency. For an
EVM-only deployment account abstraction would be a strictly
cleaner home for Invariant~\ref{inv:delegation}; we chose
off-chain tokens because the wallet drives non-EVM executors
(Hyperliquid signed L1 messages, an HTTP-only analyst service)
under a single auth middleware. MPC-custody and multi-sig
alternatives attack key custody rather than scoping signer
authority and are orthogonal to the orchestration-layer
invariants.

\paragraph{Web3 agents, bridges, and concurrent-state formalisms.}
A growing set of projects ships agent-driven Web3 tooling and
recent work~\cite{autonomous_agents_blockchain,decoagent} has begun
to benchmark the space. We are not aware of published work that
formalizes wallet operations as state transitions in a
four-dimensional space, enumerates safety invariants against real
observed failure modes, \emph{and} evaluates against an on-chain
track record rather than a synthetic benchmark. Bridge surveys
\cite{sok_bridges,price_of_interoperability} taxonomize the
provider-gap failure (class 5) and the trust assumptions users pay
to move value between ledgers; our table-driven dispatcher is an
applied instance of the routing machinery these works motivate,
specialized for a probabilistic agent. Petri nets and process
calculi~\cite{petri1962kommunikation,milner1999communicating,
puhlmann2009pi} give vocabulary for concurrent state transitions
with preservation constraints; embedding the 4D space into a
Petri-net or $\pi$-calculus formalization would let plan-level
guarantees (Invariant~\ref{inv:plan-retry}) follow structurally
rather than from a named runtime invariant, and remains future
work. HD-wallet derivation
trees~\cite{post_quantum_bitcoin} are only lightly formalized in
the academic literature, with recent post-quantum work treating
the tree as a capability-like structure under hot- and cold-
storage threat models.

%% file: content/07_discussion.tex
\section{Discussion}
\label{sec:discussion}

\paragraph{Limitations.}
The work treats the user as trusted and single-tenant; multi-user
deployment introduces adversarial-prompt surface the invariants do
not address. The model provider's API is a black box, so a silent
change to tokenization, system-prompt handling, or sampling
defaults could invalidate the reproducibility regime
(\S\ref{sec:methodology}) without an observable signal on the
wallet side. The controlled-experiment sample ($N=60$ per cell) is
small for an academic claim: the 95\% bootstrap intervals reported
are unpaired (conservative for our own with-vs-without comparisons
where a paired test e.g.\ McNemar would be more powerful) and we
do not report a per-seed variance study. Within that suite the
gpt-5-nano column is further methodology-bounded ($N=19$--$22$
because of per-turn agent-loop wall-clock; see
\S\ref{sec:evaluation}). The code-gate invariant ablations
(I\textsubscript{3}, I\textsubscript{5}, I\textsubscript{6},
I\textsubscript{7}) are exercised by in-process drivers that
mirror the production code paths rather than by toggling the gate
inside the broadcasting executor over a live trace; this isolates
the gate's logic from the broadcast side effect but does not
exercise the gate against a population of traffic-mixed inputs,
which a trace-replay row over the production audit log would.
Invariant~\ref{inv:preflight} additionally inherits a trust
assumption on the L1 RPC provider; light-client verification
(header sync + inclusion proof) would close the surface and is
the most-pressing infrastructure-level future work.

\paragraph{Out-of-scope environmental losses.}
Three losses the seven invariants do not directly defend against
sit in the threat-model bucket \emph{partly in scope} from
\S\ref{sec:background}:

\begin{description}\setlength\itemsep{2pt}
  \item[Indirect prompt injection] through avatar-ingested feeds
    (RSS, social, third-party trigger payloads):
    Invariant~\ref{inv:delegation}'s scope bounding limits the
    \emph{authority} an injected prompt can exercise but not its
    exercise on in-scope tools.
  \item[Avatar-service compromise] is a single-point-of-compromise
    distinct from executor compromise: the avatar holds
    long-lived delegation tokens across the fleet, and a breach
    would let an attacker exercise every token's scope at once.
    Mitigation today is operational (scope-bound tokens, audit
    log per call, revocation without session restart); on-chain
    session-key custody via account abstraction
    (\S\ref{sec:related}) is the natural next step.
  \item[MEV] (sandwich, JIT-liquidity, oracle / bridge
    front-running) is a counterparty-layer
    extraction~\cite{daian2020flash,qin2022quantifying_mev}
    operating on otherwise-valid signed transactions, orthogonal
    to the agent invariants. Operational mitigations
    (limit-order discipline on Trade primitives, router
    preference for on-chain auction mechanics on Bridge
    primitives, tight slippage tolerance) reduce realized loss
    on small notionals but do not close the surface.
\end{description}

\paragraph{Future work --- formal.}
The most load-bearing direction is plan-level transactionality.
Invariant~\ref{inv:plan-retry} states a verify-then-retry rule
enforced at the avatar tier, but it guards retries only
\emph{negatively}: it refuses an unsafe retry without formalizing
when forward progress on a partial plan is admissible. The
cleaner treatment is to model a plan as a compound transition
$\tau_n \circ \cdots \circ \tau_1$ in $\mathcal{S}$ with an
explicit transactionality predicate, so plan-level guarantees
fall out of the formalism rather than from a named runtime
invariant; a swap-then-bridge or deposit-then-withdraw sequence
has atomicity requirements Invariant~\ref{inv:plan-retry} does
not capture. A secondary direction is extending $\tau$ to a
partial function $\tau_{(\textsf{actor}, \Sigma)}$ parameterized
by the caller's scope set, lifting
Invariant~\ref{inv:delegation} from runtime check to algebra
(building on the object-capability vocabulary discussed in
\S\ref{sec:related}).

\paragraph{Future work --- empirical.}
The avatar fleet (\S\ref{sec:multi-tier}) is the largest source
of non-adversarial agent-issued traffic the wallet sees. Once
\texttt{audit\_log.caller\_type} is populated end-to-end, every
daily run is a labeled data point and the false-positive rate of
each invariant becomes a continuous deployment metric---an
unexpected lift in, say, Invariant~\ref{inv:preflight}'s
false-positive rate signals indexer drift or an avatar template's
intent grammar moving out of distribution, a qualitatively
different empirical posture from the scripted adversarial suite.
A trace-replay version of the I\textsubscript{3},
I\textsubscript{5}, I\textsubscript{6}, I\textsubscript{7}
ablation rows over the same audit log would replace the
in-process drivers with traffic-mixed inputs and is the cleanest
single piece of empirical follow-up.

%% file: content/08_conclusion.tex
\section{Conclusion}
\label{sec:conclusion}

We argued that a broad class of agent-orchestrated actions on
public ledgers are naturally state transitions in a
four-dimensional space, and that the point of saying so is not
description but decidability: it makes ``what was promised'' and
``what was executed'' the same kind of finite object. On that
basis we proved \emph{execution fidelity}---a session's realized
effect is either nothing or exactly the rendered transition,
exactly once---and derived seven safety invariants from the
condition rather than from the bugs that prompted them. The
guarantee stops short of the question everyone wants answered,
whether the rendered transition is what the user meant; what it
does is make that the only remaining question, over an object
small enough to check. The empirical side is an on-chain track
record of 108 production write operations plus a controlled
$N{=}60$ adversarial suite on which the full stack lifts pass
rate by $\sim$74 percentage points over a naive-ReAct baseline on
two write-aggressive backing models---and by $\sim$3 points on a
write-cautious one, which is a caution about how safety layers
are evaluated at least as much as a result about this one. The cross-tier
invariant is load-bearing in deployment: a separate avatar
service holds scoped delegation tokens against subsets of the
wallet's authority and drives the same user-facing surface
humans use, without any privileged backdoor. The cleanest single
piece of formal future work the paper suggests is lifting
Invariant~\ref{inv:plan-retry} from a runtime check into a
plan-level transactionality predicate on the compound transition
$\tau_n \circ \cdots \circ \tau_1$ (\S\ref{sec:discussion}). The
formalism and invariants apply to any probabilistic agent acting
on irreversible external state; public ledgers are a
particularly demanding testbed because every failure is publicly
and permanently observable.

%% file: references.bib
@inproceedings{yao2023react,
  author    = {Shunyu Yao and Jeffrey Zhao and Dian Yu and Nan Du and Izhak Shafran and Karthik Narasimhan and Yuan Cao},
  title     = {ReAct: Synergizing Reasoning and Acting in Language Models},
  booktitle = {International Conference on Learning Representations (ICLR)},
  year      = {2023},
  url       = {https://arxiv.org/abs/2210.03629}
}

@inproceedings{shinn2023reflexion,
  author    = {Noah Shinn and Federico Cassano and Edward Berman and Ashwin Gopinath and Karthik Narasimhan and Shunyu Yao},
  title     = {Reflexion: Language Agents with Verbal Reinforcement Learning},
  booktitle = {Advances in Neural Information Processing Systems (NeurIPS)},
  year      = {2023},
  url       = {https://arxiv.org/abs/2303.11366}
}

@misc{mcp2024,
  author       = {{Anthropic}},
  title        = {Model Context Protocol Specification},
  year         = {2024},
  howpublished = {\url{https://modelcontextprotocol.io/specification}},
  note         = {Accessed 2026-04-19}
}

@inproceedings{liu2023lost,
  author    = {Nelson F. Liu and Kevin Lin and John Hewitt and Ashwin Paranjape and Michele Bevilacqua and Fabio Petroni and Percy Liang},
  title     = {Lost in the Middle: How Language Models Use Long Contexts},
  booktitle = {Transactions of the Association for Computational Linguistics (TACL)},
  year      = {2024},
  url       = {https://arxiv.org/abs/2307.03172}
}

@misc{laban2025llms,
  author       = {Philippe Laban and Hiroaki Hayashi and Yingbo Zhou and Jennifer Neville},
  title        = {LLMs Get Lost In Multi-Turn Conversation},
  year         = {2025},
  howpublished = {arXiv:2505.06120},
  url          = {https://arxiv.org/abs/2505.06120}
}

@article{ji2023survey,
  author  = {Ziwei Ji and Nayeon Lee and Rita Frieske and Tiezheng Yu and Dan Su and Yan Xu and Etsuko Ishii and Ye Jin Bang and Andrea Madotto and Pascale Fung},
  title   = {Survey of Hallucination in Natural Language Generation},
  journal = {ACM Computing Surveys},
  volume  = {55},
  number  = {12},
  year    = {2023},
  doi     = {10.1145/3571730}
}

@misc{huang2023survey,
  author       = {Lei Huang and Weijiang Yu and Weitao Ma and Weihong Zhong and Zhangyin Feng and Haotian Chen and Qianglong Chen and Weihua Peng and Xiaocheng Feng and Bing Qin and Ting Liu},
  title        = {A Survey on Hallucination in Large Language Models: Principles, Taxonomy, Challenges, and Open Questions},
  year         = {2023},
  howpublished = {arXiv:2311.05232},
  url          = {https://arxiv.org/abs/2311.05232}
}

@book{gray1993transaction,
  author    = {Jim Gray and Andreas Reuter},
  title     = {Transaction Processing: Concepts and Techniques},
  publisher = {Morgan Kaufmann},
  year      = {1993}
}

@article{herlihy1990linearizability,
  author  = {Maurice P. Herlihy and Jeannette M. Wing},
  title   = {Linearizability: A Correctness Condition for Concurrent Objects},
  journal = {ACM Transactions on Programming Languages and Systems},
  volume  = {12},
  number  = {3},
  year    = {1990},
  doi     = {10.1145/78969.78972}
}

@inproceedings{daian2020flash,
  author    = {Philip Daian and Steven Goldfeder and Tyler Kell and Yunqi Li and Xueyuan Zhao and Iddo Bentov and Lorenz Breidenbach and Ari Juels},
  title     = {Flash Boys 2.0: Frontrunning in Decentralized Exchanges, Miner Extractable Value, and Consensus Instability},
  booktitle = {IEEE Symposium on Security and Privacy (S\&P)},
  year      = {2020},
  doi       = {10.1109/SP40000.2020.00040}
}

@inproceedings{zhou2023sok,
  author    = {Liyi Zhou and Xihan Xiong and Jens Ernstberger and Stefanos Chaliasos and Zhipeng Wang and Ye Wang and Kaihua Qin and Roger Wattenhofer and Dawn Song and Arthur Gervais},
  title     = {{SoK}: Decentralized Finance (DeFi) Attacks},
  booktitle = {IEEE Symposium on Security and Privacy (S\&P)},
  year      = {2023},
  url       = {https://arxiv.org/abs/2208.13035}
}

@article{saltzer1984end,
  author  = {J. H. Saltzer and D. P. Reed and D. D. Clark},
  title   = {End-to-End Arguments in System Design},
  journal = {ACM Transactions on Computer Systems},
  volume  = {2},
  number  = {4},
  year    = {1984},
  doi     = {10.1145/357401.357402}
}

@phdthesis{petri1962kommunikation,
  author = {Carl Adam Petri},
  title  = {Kommunikation mit Automaten},
  school = {Technische Universit{\"a}t Darmstadt},
  year   = {1962}
}

@book{milner1999communicating,
  author    = {Robin Milner},
  title     = {Communicating and Mobile Systems: The {$\pi$}-calculus},
  publisher = {Cambridge University Press},
  year      = {1999}
}

@misc{sok_bridges,
  title        = {{SoK}: {Security} of Cross-Chain Bridges: {Attack} Surfaces, Defenses, and Open Problems},
  author       = {Mengya Zhang and Xiaokuan Zhang and Josh Barbee and Yinqian Zhang and Zhiqiang Lin},
  year         = {2023},
  howpublished = {arXiv:2312.12573},
  url          = {https://arxiv.org/abs/2312.12573}
}

@misc{autonomous_agents_blockchain,
  title        = {Autonomous Agents on Blockchains: {Standards}, Execution Models, and Trust Boundaries},
  author       = {Saad Alqithami},
  year         = {2026},
  howpublished = {arXiv:2601.04583},
  url          = {https://arxiv.org/abs/2601.04583}
}

@incollection{puhlmann2009pi,
  title        = {A Look Around the Corner: {The} $\pi$-Calculus},
  author       = {Frank Puhlmann and Mathias Weske},
  year         = {2009},
  booktitle    = {Transactions on {Petri} Nets and Other Models of Concurrency {II}},
  publisher    = {Springer},
  doi          = {10.1007/978-3-642-00899-3_4}
}

@misc{price_of_interoperability,
  title        = {The Price of Interoperability: {Exploring} Cross-Chain Bridges and Their Economic Consequences},
  author       = {Yiyue Cao and Mingzhe Zheng and Lin William Cong and Siguang Li and Xuechao Wang},
  year         = {2026},
  howpublished = {arXiv:2604.03083},
  url          = {https://arxiv.org/abs/2604.03083}
}

@inproceedings{lost_beginning_reasoning,
  title        = {Lost at the Beginning of Reasoning},
  author       = {Baohao Liao and Xinyi Chen and Sara Rajaee and Yuhui Xu and Christian Herold and Anders S{\o}gaard and Maarten de Rijke and Christof Monz},
  year         = {2025},
  howpublished = {arXiv:2506.22058},
  url          = {https://arxiv.org/abs/2506.22058}
}

@misc{intent_mismatch_lost,
  title        = {Intent Mismatch Causes {LLMs} to Get Lost in Multi-Turn Conversation},
  author       = {Geng Liu and Fei Zhu and Rong Feng and Changyi Ma and Shiqi Wang and Gaofeng Meng},
  year         = {2026},
  howpublished = {arXiv:2602.07338},
  url          = {https://arxiv.org/abs/2602.07338}
}

@article{post_quantum_bitcoin,
  title        = {Towards Post-Quantum {Bitcoin} Blockchain using {Dilithium} Signature},
  author       = {Iyane Seck and Adeline Roux-Langlois},
  year         = {2025},
  journal      = {{IACR} Communications in Cryptology ({CIC})},
  volume       = {2},
  number       = {3},
  doi          = {10.62056/ak5wom2hd},
  url          = {https://cic.iacr.org/p/2/3/3}
}

@article{decoagent,
  title        = {{DeCoAgent}: {Large} Language Model Empowered Decentralized Autonomous Collaboration Agents Based on Smart Contracts},
  author       = {Anan Jin and Yuhang Ye and Brian Lee and Yuansong Qiao},
  year         = {2024},
  journal      = {{IEEE} Access},
  volume       = {12},
  pages        = {155234--155245},
  doi          = {10.1109/ACCESS.2024.3481641}
}

@misc{mcp_sok,
  title        = {Systematization of Knowledge: {Security} and Safety in the {Model Context Protocol} Ecosystem},
  author       = {Shiva Gaire and Srijan Gyawali and Saroj Mishra and Suman Niroula and Dilip Thakur and Umesh Yadav},
  year         = {2025},
  howpublished = {arXiv:2512.08290},
  url          = {https://arxiv.org/abs/2512.08290}
}

@misc{wu2023autogen,
  title        = {{AutoGen}: Enabling Next-Gen {LLM} Applications via Multi-Agent Conversation},
  author       = {Qingyun Wu and Gagan Bansal and Jieyu Zhang and Yiran Wu and Beibin Li and Erkang Zhu and Li Jiang and Xiaoyun Zhang and Shaokun Zhang and Jiale Liu and Ahmed Hassan Awadallah and Ryen W. White and Doug Burger and Chi Wang},
  year         = {2023},
  howpublished = {arXiv:2308.08155},
  url          = {https://arxiv.org/abs/2308.08155}
}

@misc{hong2023metagpt,
  title        = {{MetaGPT}: Meta Programming for A Multi-Agent Collaborative Framework},
  author       = {Sirui Hong and Mingchen Zhuge and Jiaqi Chen and Xiawu Zheng and Yuheng Cheng and Ceyao Zhang and Jinlin Wang and Zili Wang and Steven Ka Shing Yau and Zijuan Lin and Liyang Zhou and Chenyu Ran and Lingfeng Xiao and Chenglin Wu and J{\"u}rgen Schmidhuber},
  year         = {2023},
  howpublished = {arXiv:2308.00352},
  url          = {https://arxiv.org/abs/2308.00352}
}

@book{levy1984capability,
  author       = {Henry M. Levy},
  title        = {Capability-Based Computer Systems},
  publisher    = {Digital Press},
  year         = {1984},
  url          = {https://homes.cs.washington.edu/~levy/capabook/}
}

@phdthesis{miller2006robust,
  author       = {Mark S. Miller},
  title        = {Robust Composition: {Towards} a Unified Approach to Access Control and Concurrency Control},
  school       = {Johns Hopkins University},
  year         = {2006},
  url          = {http://www.erights.org/talks/thesis/markm-thesis.pdf}
}

@techreport{miller2003capabilities,
  author       = {Mark S. Miller and Ka-Ping Yee and Jonathan Shapiro},
  title        = {Capability Myths Demolished},
  institution  = {Systems Research Laboratory, Johns Hopkins University},
  number       = {SRL2003-02},
  year         = {2003},
  url          = {https://srl.cs.jhu.edu/pubs/SRL2003-02.pdf}
}

@misc{erc4337,
  author       = {Vitalik Buterin and Yoav Weiss and Dror Tirosh and Shahaf Nacson and Alex Forshtat and Kristof Gazso and Tjaden Hess},
  title        = {{ERC-4337}: Account Abstraction Using Alt Mempool},
  howpublished = {Ethereum Improvement Proposal},
  year         = {2021},
  url          = {https://eips.ethereum.org/EIPS/eip-4337}
}

@misc{safe_aa_modules,
  author       = {{Safe Ecosystem Foundation}},
  title        = {{Safe} Smart Account Modules and Session Keys},
  howpublished = {Technical documentation},
  year         = {2024},
  url          = {https://docs.safe.global/advanced/smart-account-modules}
}

@misc{argent_session_keys,
  author       = {{Argent Labs}},
  title        = {Session Keys in {Argent}: scope-bounded delegated signers},
  howpublished = {Technical blog},
  year         = {2023},
  url          = {https://www.argent.xyz/blog/session-keys/}
}

@inproceedings{qin2022quantifying_mev,
  author       = {Kaihua Qin and Liyi Zhou and Arthur Gervais},
  title        = {Quantifying Blockchain Extractable Value: How Dark is the Forest?},
  booktitle    = {{IEEE} Symposium on Security and Privacy ({S}\&{P})},
  year         = {2022},
  pages        = {198--214},
  doi          = {10.1109/SP46214.2022.9833734}
}
